\documentclass[journal]{IEEEtran}
\usepackage{amsmath,amsfonts,amsthm}
\usepackage[ruled,commentsnumbered,linesnumbered]{algorithm2e}
\usepackage{array}
\usepackage{textcomp}
\usepackage[dvipsnames]{xcolor}
\usepackage{stfloats}
\ifCLASSOPTIONcompsoc
\usepackage[caption=false,font=normalsize,labelfont=sf,textfont=sf]{subfig}
\else
\usepackage[caption=false,font=footnotesize]{subfig}
\fi
\usepackage{url}
\usepackage{verbatim}
\usepackage{graphicx}
\usepackage{cite}
\usepackage[colorlinks,linkcolor=BrickRed,citecolor=Aquamarine,urlcolor=cyan]{hyperref} % blue
\hypersetup{
  colorlinks=true,
  linkcolor=blue,
  citecolor=blue,
  urlcolor=blue
}
\usepackage{makecell}
\usepackage{multirow}
\usepackage{color}
\usepackage{booktabs}
\usepackage{amsmath}

\allowdisplaybreaks[4]
\usepackage{amssymb}
\usepackage{threeparttable}

\theoremstyle{remark}
\newtheorem{remark}{Remark}
\newtheorem{proposition}{Proposition}
\newtheorem{lemma}{Lemma}

\usepackage{orcidlink} 

\begin{document}

\title{Achievable-Rate Analysis of MISO Systems with Transmit-Side Multiport Matching Networks}
% \title{A Generalization of Achievable Rate for MISO Systems with Multiport Matching Networks}
\author{Wendong Cheng, Li Chen and Weidong Wang
        % <-this % stops a space
% \thanks{This work was supported by the Natural Science Foundation of China (Grant No. 62522126) and Anhui Provincial Natural Science Foundation (No. 2308085J24). \emph{(Corresponding author: Li Chen.)}}
\thanks{Wendong Cheng, Li Chen and Weidong Wang are with the CAS Key Laboratory of Wireless Optical Communication, University of Science and Technology of China (USTC), Hefei 230027, China (e-mail: cwd01@mail.ustc.edu.cn; chenli87@ustc.edu.cn; wdwang@ustc.edu.cn).}
}



\maketitle
\begin{abstract} % because of their capability for joint decoupling and matching
Characterizing communication performance under the physical constraints imposed by radio frequency front-end circuits is essential for bridging communication-theoretic analysis and practical circuit design. In this work, we investigate the achievable-rate upper bound of a multiple-input single-output (MISO) system with a transmitter-side interconnected multiport matching network (MMN) under multiport Bode--Fano constraints and develop a rate-oriented MMN circuit realization method. First, based on a circuit-theoretic communication model and the concept of directional gain from multivariable control theory, we reveal how the directional characteristics of MMN power transmission interact with wireless propagation to affect communication performance. Then, building on this directional-gain interpretation, we reformulate the matrix-valued functional optimization problem that characterizes the achievable-rate upper bound and derive the optimal transmission-coefficient structure. Further, a greedy constraint-wise repair method is developed to obtain a feasible suboptimal solution. Finally, through a rate-oriented MMN circuit realization method, we demonstrate that the derived theoretical insights provide effective guidance for practical MMN design. Numerical results validate the theoretical analysis of the achievable-rate upper bound and the effectiveness of the proposed MMN circuit realization method.
% Impedance matching network design plays a crucial role in multi-antenna systems. To characterize the mechanism by which transmit-side interconnected MMNs affect the communication performance of digital beamforming multiple-input single-output (DBF-MISO) systems, in this paper, we propose an analytical framework based on circuit theory and multivariable control theory. First, an equivalent channel model incorporating the MMN is established. Then, leveraging the concept of directional gain, we reveal that the equivalent channel gain is governed jointly by the overall power transmission and the alignment between the propagation vector and the dominant transmission subspace. Furthermore, we derive two physically interpretable necessary conditions for channel gain maximization. Numerical results validate the effectiveness of the proposed analytical framework.
\end{abstract}


\section{Introduction}\label{section_1}
\IEEEPARstart {I}{mpedance} matching is a key consideration in antenna systems \cite{cui2023absorptive} and radio frequency (RF) front-end design \cite{cheng2024multimode}. Impedance matching networks determine how efficiently power is transferred from the RF chains to the antennas. Impedance mismatch degrades the realized array gain \cite{soltani2024wide} and can further affect system-level communication performance, including the achievable data rate \cite{lau2006impact}. Therefore, achieving impedance matching in practical communication systems is highly desirable and has attracted sustained attention from both the RF and communications communities.

Classical RF impedance matching was initially studied within a single-source single-load framework, in which maximum power transfer is achieved under conjugate-matching condition \cite{abdallah2016maximum}. Based on this principle, the authors in \cite{putra2015impedance} designed an L-network for a monopole antenna to minimize the voltage standing wave ratio, while the authors in \cite{choi2020adaptive} used a half-wavelength transformer to reduce the antenna reflection coefficient. Since minimizing the magnitude of the reflection coefficient does not necessarily maximize the power delivered to the antenna when the matching network is lossy, the authors in \cite{rahola2011optimization} adopted transducer power gain as a more direct metric for power-transfer performance.

With the development of multi-antenna systems, impedance matching has extended from single-source single-load power transfer to multi-source multi-load power transfer. Under this multiport framework, maximum power transfer is achieved when the load impedance matrix equals the Hermitian transpose of the source impedance matrix \cite{desoer1973maximum}. In antenna arrays, mutual coupling makes the port impedances interdependent \cite{li2021decoupling}, which in turn requires joint shaping of the load impedance matrix and increases the complexity of the matching network. For coupled antennas, the authors in \cite{broyde2014some} showed that independent single-port matching lacks full tuning capability, whereas interconnected multiport matching networks (MMNs) can realize multiport conjugate matching. Along this direction, the authors in \cite{nie2014systematic} developed a passive lossless MMN realization method that achieves source-side zero reflection with minimum circuit complexity at the design frequency. To balance power-transfer performance and implementation complexity, the authors in \cite{shen2015impedance} further proposed sparse and ladder-type multiport matching structures for compact antenna arrays. Overall, these RF studies investigate impedance matching in multi-antenna systems from the perspective of efficient power transfer from the RF chains to the antenna ports. %To maximize the average received power of antenna arrays, the authors in \cite{shen2015impedance} developed sparse optimal and lower-complexity ladder-type multiport matching methods. % Mutual coupling makes the antenna port impedances interdependent and significantly complicates the matching problem.

In addition to power-transfer efficiency, several studies have investigated impedance matching using communication performance metrics, such as the end-to-end signal-to-noise ratio (SNR) \cite{sarabandi2022bandwidth} and achievable rate. These metrics depend jointly on the matching network, the propagation channel, the array scattering characteristics \cite{li2022study}, and the beamforming or combining strategy. Circuit-theoretic communication models \cite{ivrlavc2010toward, wallace2004mutual} provide a natural framework for incorporating antenna and matching-network effects into communication performance analysis. For compact receive arrays, the authors in \cite{fei2008optimal} derived the optimal single-port matching impedance for maximizing the capacity of a narrowband \(2\times2\) multiple-input multiple-output (MIMO) system. For wideband single-input multiple-output (SIMO) and multiple-input single-output (MISO) systems, the authors in \cite{saab2021optimizing} optimized the parameters of fixed per-antenna single-port matching networks (SMNs) to maximize the achievable rate. For broadband MIMO receivers employing orthogonal frequency-division multiplexing, the authors in \cite{kundu2017enhancing} employed a Butler matrix for array decoupling and optimized per-port T-networks to maximize capacity. In summary, these studies focus on improving communication performance through impedance matching under specific matching-network structures.

Recently, several works have moved beyond specific matching-network topologies to investigate the fundamental achievable-rate limits imposed by physically realizable impedance matching networks. Bode--Fano broadband matching theory \cite{fano1950theoretical} characterizes the fundamental tradeoff between matching performance and bandwidth for passive matching networks. Based on this limitations, the authors in \cite{deshpande2023achievable} derived the achievable-rate upper bound for single-input single-output (SISO) systems with a transmit-side SMN and developed a practical SMN parameter design method to approach this upper bound. This framework was further extended in \cite{deshpande2024generalization} to single-RF-chain MISO systems under the same matching architecture, where the achievable-rate upper bound was derived and the received SNR was shown to be proportional to the overall power transfer efficiency of the transmit-side SMN. %For transmit-side single-port matching, the authors in \cite{deshpande2023achievable} incorporated the Bode--Fano broadband matching constraint to characterize the achievable-rate upper bound of SISO systems. Furthermore, for MISO systems with a single RF chain and a transmit-side single-port matching network, the authors in \cite{deshpande2024generalization} incorporated the Bode--Fano broadband matching constraint to characterize the achievable-rate upper bound of MISO systems and demonstrated that the SNR is proportional to the overall power transfer efficiency of the transmit-side two-port matching network. %Furthermore, for MISO systems with a single RF chain, the authors in \cite{deshpande2024generalization} showed that the transmit-side single-port matching network affects the SNR through its transmission coefficient and optimized the transmission response to maximize the achievable rate. 

% Although the above studies introduce communication metrics into matching network design for multi-antenna systems, the corresponding optimization frameworks still fundamentally rely on single-source single-load matching.
However, existing achievable-rate upper-bound analyses are limited to single-source single-load matching architectures. Modern multi-antenna transmitters commonly employ multiple active RF chains \cite{barneto2022beamformer} to support flexible beamforming and spatial transmission, resulting in a multi-source multi-load configuration. In this configuration, interconnected MMNs provide the full tuning capability required for multiport matching. From a communication-theoretic perspective, characterizing the achievable-rate upper bound for multi-antenna systems with transmit-side interconnected MMNs under multiport Bode--Fano constraints \cite{nie2016bandwidth} remains an open issue. 

The challenges are twofold. First, characterizing the achievable-rate upper bound requires a clear understanding of how an interconnected MMN affects communication performance. Unlike an SMN, an MMN provides different power-transfer efficiencies for different excitation directions. Meanwhile, the propagation channel weights these directions differently in forming the received signal. Quantifying how these two directional effects interact to determine the equivalent channel gain constitutes the first challenge. Second, the achievable-rate upper-bound problem under multiport Bode--Fano constraints takes the form of a matrix-valued functional optimization problem, making it challenging to solve. % the mechanism by which interconnected MMNs affect communication performance remains insufficiently understood.Therefore, the RF transmission coefficient alone cannot fully characterize the resulting SNR. % Multi-port matching networks need to ensure that the optimal space pointing to the corresponding source-side excitation has no power transmission loss.

% multiple-input single-output (DBF-MISO) systems with transmitter-side interconnected MMNs First, based on a circuit-theoretic communication model, we derive the equivalent channel that incorporates the transmitter-side MMN. Then, by introducing the concepts of input direction and directional gain from multivariable control theory \cite{skogestad2005multivariable}, we reveal that the equivalent channel gain is jointly governed by the overall power transmission of the MMN and the alignment between the propagation vector and its dominant transmission subspace.
To address these challenges, this paper develops a directional-gain-based framework for analyzing and a tractable solution method for characterizing the Bode--Fano-constrained achievable-rate upper bound of MISO systems with transmit-side MMNs. First, based on a circuit-theoretic communication model and the concepts of input direction and directional gain from multivariable control theory \cite{skogestad2005multivariable}, we reveal that the equivalent channel gain is formed by projecting the power-normalized propagation vector onto the input directions of the MMN-induced transfer matrix and weighting these projections by the corresponding directional gains. Building on this interpretation, we reformulate the matrix-valued functional optimization problem as a scalar transmission-allocation problem and develop a greedy constraint-wise repair method to obtain a feasible suboptimal solution. Finally, we develop a rate-oriented circuit realization method that directly optimizes passive multiport circuit parameters under prescribed topologies. The main contributions are summarized as follows.%First, a circuit-theoretic equivalent channel model is established. Then, using directional gain from multivariable control theory, we demonstrate that the equivalent channel gain depends on overall power transmission and the alignment between the propagation vector and the dominant transmission subspace. Furthermore, we provide a physical interpretation of channel gain maximization through MMN design. Numerical results are provided to validate the effectiveness of the proposed analytical framework.
\begin{itemize}
    \item \textbf{Directional-gain characterization of MMN-induced channel shaping.} Based on a circuit-theoretic communication model, we derive the equivalent MISO channel incorporating the transmit-side MMN and represent it as a linear transfer-matrix mapping. Using the concepts of input direction and directional gain, we show that the equivalent channel gain is jointly determined by the overall RF power transmission and the alignment between the power-normalized propagation vector and the dominant transmission subspace of the resulting transfer matrix. This result identifies directional-gain shaping as an additional degree of freedom for enhancing communication performance without increasing the overall RF power transmission.

    \item \textbf{Achievable-rate upper-bound characterization under multiport Bode--Fano constraints.} Building on the above directional-gain insight, we reformulate the original matrix-valued functional optimization problem as a scalar transmission-allocation problem and derive the optimal structure of the transmission coefficient using variational calculus and the Karush-Kuhn-Tucker (KKT) conditions. The resulting analysis reveals that the ideal-matching SNR can be attained without perfect impedance matching and that limited matching resources should be allocated across frequency according to rate benefit and Bode--Fano matching cost. We further develop a greedy constraint-wise repair method to obtain a feasible suboptimal solution.

    \item \textbf{Rate-oriented passive MMN circuit realization.} Since a smaller fitting error to the theoretically optimal MMN response derived under the Bode--Fano constraints does not necessarily yield a higher achievable rate, we formulate a passive MMN circuit realization problem that directly maximizes the achievable rate under a prescribed multiport topology. We then develop a multi-start gradient-based optimization method with structured initializations and analytical gradients to determine the circuit parameters, with each branch admittance parameterized in Foster form. Numerical results show that the optimized circuit responses follow the trend implied by the theoretical optimum, demonstrating the practical relevance of the derived theoretical insights.
\end{itemize}

The remainder of this paper is organized as follows. Section \ref{section_2} derives the equivalent channel and formulates the achievable-rate upper-bound problem under multiport Bode--Fano constraints. In Section \ref{section_3}, we interpret the coupling between multiport power transmission and propagation using directional gain. Section \ref{section_4} derives the achievable-rate upper bound and presents the greedy constraint-wise repair method. The rate-oriented passive MMN circuit realization method is developed in Section \ref{section_5}. Section \ref{section_6} provides numerical results to validate the proposed theoretical analysis and circuit realization method. Finally, Section \ref{section_7} concludes the paper.

\emph{Notation:}
Bold uppercase letters denote matrices and bold lowercase letters denote vectors.
For a matrix $\mathbf A$, $\mathbf A^T$, $\mathbf A^H$, $\mathbf A^*$, $\mathbf A^{-1}$, and $\mathbf A^{-H}$ denote its transpose, Hermitian transpose, complex conjugate, inverse, and inverse Hermitian transpose, respectively.
The notation $[\mathbf A]_{ij}$ denotes the $(i,j)$-th entry of $\mathbf A$.
$\operatorname{tr}(\mathbf A)$ denotes the trace of $\mathbf A$, and
$\operatorname{diag}(x_1,\ldots,x_N)$ denotes the diagonal matrix with
$x_1,\ldots,x_N$ on its main diagonal.
For a vector $\mathbf a$, $\|\mathbf a\|$ denotes its Euclidean norm.
For a scalar $x$, $|x|$ denotes its magnitude.
$\mathbb E[\cdot]$ denotes expectation, and $\Re\{x\}$ denotes the real part of $x$.
The symbols $\mathbb R$ and $\mathbb C$ denote the sets of real and complex numbers, respectively.
$\mathbf I_N$ and $\mathbf 0_{m\times n}$ denote the $N\times N$ identity matrix and the $m\times n$ zero matrix, respectively. For Hermitian matrices $\mathbf A$ and $\mathbf B$, $\mathbf A\succeq\mathbf B$ and $\mathbf A\succ\mathbf B$ indicate that $\mathbf A-\mathbf B$ is positive semidefinite and positive definite, respectively. The operators $[x]_a^b$ and $[x]_+$ are defined as $[x]_a^b=\min\{\max\{x,a\},b\}$ and $[x]_+=\max\{x,0\}$, respectively. $\{i\}_1^N$ is shorthand for $i = \{1, 2, \dots, N \}$.

% Circuit-Theoretic Modeling of DBF-MISO Systems a transmit-side MMN
\section{System Model and Problem Formulation}\label{section_2}
In this section, we first establish a circuit-theoretic model for a MISO system with multiple RF chains. Then, we derive the achievable rate as a function of the transmit-side MMN parameters. Finally, we formulate the achievable-rate maximization problem under multiport Bode--Fano constraints and identify the key challenges in solving it.
\subsection{Circuit-Theoretic Modeling of MISO Systems}\label{section_2_1}
Fig.~\ref{system_model} depicts the circuit-theoretic model of the considered MISO system. On the transmit side, a passive and lossless MMN connects the \(M\) RF chains to the \(M\) mutually coupled transmit antennas. On the receive side, the single receive load is assumed to be perfectly matched.

Root power waves are adopted as the representations of the signals throughout the analysis. At the source-side ports of the MMN, the incident and reflected root power wave vectors are denoted by $\mathbf{a}_{\mathrm{T}}(f)\in \mathbb{C}^{M\times1}$ and $\mathbf{b}_{\mathrm{T}}(f) \in \mathbb{C}^{M\times1}$, respectively, and are defined as
\begin{equation*}%\label{eq:atf_btf}
\mathbf{a}_{\mathrm{T}}(f)
= \frac{\mathbf{v}_{\mathrm{T}}(f)+Z_{0}\mathbf{i}_{\mathrm{T}}(f)}
{2\sqrt{\Re{(Z_{0})}}},
\quad
\mathbf{b}_{\mathrm{T}}(f)
= \frac{\mathbf{v}_{\mathrm{T}}(f)-Z_{0}^{*}\mathbf{i}_{\mathrm{T}}(f)}
{2\sqrt{\Re{(Z_{0})}}},
\end{equation*}
% \begin{subequations}
% \begin{align}
% \mathbf{a}_{\mathrm{T}}(f) = \frac{\mathbf{v}_{\mathrm{T}}(f)+Z_{0}\mathbf{i}_{\mathrm{T}}(f)}{2\sqrt{\Re{(Z_{0})}}}, \label{eq:atf} \\
% \mathbf{b}_{\mathrm{T}}(f) = \frac{\mathbf{v}_{\mathrm{T}}(f)-Z_{0}^{*}\mathbf{i}_{\mathrm{T}}(f)}{2\sqrt{\Re{(Z_{0})}}}, \label{eq:btf}
% \end{align}
% \end{subequations}
where $Z_{0}$ denotes the reference impedance, $\mathbf{v}_{\mathrm{T}}(f) \in \mathbb{C}^M$ and $\mathbf{i}_{\mathrm{T}}(f) \in \mathbb{C}^M$ denote the windowed
Fourier transforms of the random voltage vector $\mathbf{v}_{\mathrm{T}}(t) \in \mathbb{C}^{M\times1}$ and the current vector $\mathbf{i}_{\mathrm{T}}(t) \in \mathbb{C}^{M\times1}$ at the source-side ports of the MMN, respectively. Accordingly, the total transmit power spectral density at the source-side ports of the MMN is given by
\begin{equation}\label{eq:PT}
    P_\mathrm{T}(f)=\operatorname{tr}\left( \mathbf{P}_{\mathbf{a}_{\mathrm{T}}}(f) \right)=\lim_{T_0\rightarrow \infty} \frac{1}{T_{0}}\mathbb{E}\left[\| \mathbf{a}_{\mathrm{T}}(f) \|^2 \right],
\end{equation}
where
\begin{equation*}
    \mathbf{P}_{\mathbf{a}_{\mathrm{T}}}(f) = \lim_{T_0\rightarrow \infty} \frac{1}{T_{0}}\mathbb{E}\left[\mathbf{a}_{\mathrm{T}}(f) \mathbf{a}_{\mathrm{T}}^H(f) \right]
\end{equation*}
is the power spectral density matrix of the incident root power wave vector, and \(T_0\) denotes the observation interval of the windowed Fourier transform \cite{deshpande2024generalization}.

\begin{figure}
\centering
\includegraphics[scale=0.65]{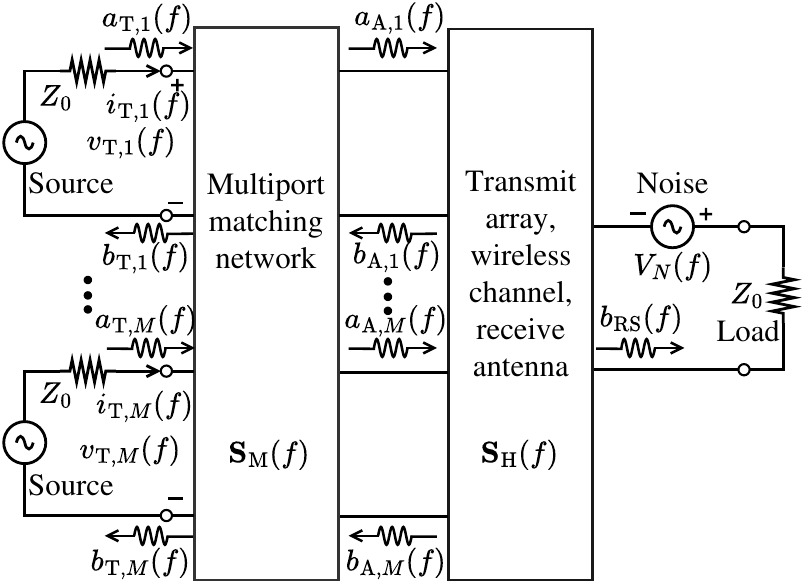}
\caption{Circuit-theoretic model of the considered MISO system with a transmit-side interconnected MMN.}\label{system_model} %图片标题
\end{figure}

At the receive load, the desired signal is represented by the scalar root power wave $b_{\mathrm{RS}}(f)$. The corresponding power spectral density is defined as
\begin{equation}\label{eq:PRS}
    P_{\mathrm{RS}}(f) = \lim_{T_0\rightarrow \infty} \frac{1}{T_{0}}\mathbb{E}\left[ \left |b_{\mathrm{RS}}(f)\right|^2 \right].
\end{equation}
Since the receive load is perfectly impedance matched, no root power wave is reflected. The background radiation noise is modeled by an equivalent noise voltage source $V_{N}(f)$, whose associated noise root power wave has power spectral density $N_{0}=k_{B}T_{n}$, where $k_{B}$ is the Boltzmann constant and $T_{n}$ is the equivalent noise temperature. Accordingly, the SNR is given by
\begin{equation}\label{eq:SNR}
    \mathrm{SNR}(f) = \frac{P_{\mathrm{RS}}(f)}{N_{0}}.
\end{equation}

\subsection{Equivalent Channel and the Corresponding SNR}\label{section_2_2}
Based on the circuit-theoretic model established above, we next incorporate the transmit-side MMN into the equivalent-channel model and derive the corresponding SNR. The transmit-side MMN is characterized by the \(2M\times 2M\) scattering matrix
\begin{equation}\label{eq:SM}
    \mathbf{S}_{\mathrm{M}}(f) = \begin{bmatrix}
    \mathbf{S}_\mathrm{M,11}(f) & \mathbf{S}_\mathrm{M,12}(f) \\
    \mathbf{S}_\mathrm{M,21}(f) & \mathbf{S}_\mathrm{M,22}(f)
    \end{bmatrix}.
\end{equation}
According to multiport network theory, the root power wave vectors at the source-side and antenna-side ports of the MMN satisfy
%the root power wave vectors at the source-side ports and the antenna-side ports satisfy
\begin{equation}\label{eq:relation_MN}
\begin{bmatrix}
\mathbf b_{\mathrm T}^{T}(f), \, \mathbf a_{\mathrm A}^{T}(f)
\end{bmatrix}^{T}
=
\mathbf S_{\mathrm M}(f)
\begin{bmatrix}
\mathbf a_{\mathrm T}^{T}(f),\, \mathbf b_{\mathrm A}^{T}(f)
\end{bmatrix}^{T}.
\end{equation}
where $\mathbf{a}_{\mathrm A}(f)\in\mathbb C^{M\times1}$ and
$\mathbf{b}_{\mathrm A}(f)\in\mathbb C^{M\times1}$ denote the root power wave vectors incident on and reflected from the transmit antenna array, respectively.

The mutually coupled transmit antenna array, the wireless propagation channel, and the receive antenna are modeled as an equivalent $(M+1)$-port network with scattering matrix
\begin{equation}\label{eq:SH}
\mathbf{S}_\mathrm{H}(f) =
\begin{bmatrix}
\mathbf{S}_\mathrm{T}(f) & \mathbf{s}_\mathrm{TR}(f) \\
\mathbf{s}_\mathrm{RT}^T(f) & S_\mathrm{R}(f)
\end{bmatrix}
\stackrel{(a)}{\approx}
\begin{bmatrix}
\mathbf{S}_\mathrm{T}(f) & \mathbf{0}_{M\times 1} \\
\mathbf{s}_\mathrm{RT}^T(f) & S_\mathrm{R}(f)
\end{bmatrix}.
\end{equation}
where $\mathbf{S}_{\mathrm{T}}(f) \in \mathbb{C}^{M\times M}$ denotes the scattering matrix of the transmit antenna array, $S_{\mathrm{R}}(f)\in \mathbb{C}$ denotes the scattering parameter of the receive antenna, and $\mathbf{s}_{\mathrm{RT}}(f)\in \mathbb{C}^{M\times1}$ denotes the scattering parameters of the wireless propagation channel. The vector \(\mathbf{s}_{\mathrm{RT}}(f)\) incorporates the effects of antenna gain and frequency-selective fading. In step (a), we adopt the unilateral approximation \cite{ivrlavc2010toward} by neglecting the reverse coupling from the receiver to the transmitter, i.e., $\mathbf{s}_{\mathrm{TR}}(f)=\mathbf{0}_{M\times1}$. This approximation is justified when the transmit and receive antennas are sufficiently separated. With the receive load perfectly matched, the corresponding root power wave satisfy
\begin{equation}\label{eq:relation_TR}
\begin{bmatrix} \mathbf{b}_{\mathrm A}^{T}(f), \, b_{\mathrm{RS}}(f) \end{bmatrix}^T = \mathbf S_\mathrm H(f) \begin{bmatrix} \mathbf a_{\mathrm A}^T(f), \, 0 \end{bmatrix}^T. 
\end{equation}

By combining (\ref{eq:relation_MN}), (\ref{eq:SH}), and (\ref{eq:relation_TR}), \(\mathbf a_{\mathrm A}(f)\) can be expressed in terms of \(\mathbf a_{\mathrm T}(f)\) as
\begin{equation}\label{eq:relation_aA_aT}
\mathbf a_{\mathrm A}(f)
=
\left(\mathbf I_M-\mathbf S_{\mathrm{M},22}(f)\mathbf S_{\mathrm T}(f)\right)^{-1}
\mathbf S_{\mathrm{M},21}(f)\mathbf a_{\mathrm T}(f).
\end{equation}
Substituting (\ref{eq:relation_aA_aT}) into (\ref{eq:relation_TR}) yields $b_{\mathrm{RS}}(f)=\mathbf h_{\mathrm{MISO}}^{H}(f)\mathbf a_{\mathrm T}(f)$,
where
\begin{equation}\label{eq:hMISO}
\mathbf h_{\mathrm{MISO}}(f) \!
=\!
\mathbf S_{\mathrm{M},21}^{H}(f) \!
\left(\mathbf I_M \!- \! \mathbf S_{\mathrm{M},22}(f)\mathbf S_{\mathrm T}(f)\right)^{-H}\!
\mathbf s_{\mathrm{RT}}^{*}(f)
\end{equation}
is the equivalent MISO channel vector relating incident root power wave vector $\mathbf a_{\mathrm T}(f)$ at the source-side ports of the MMN to the received root power wave. It captures not only the wireless propagation channel, but also the effects of frequency-selective antenna mutual coupling and the transmit-side MMN.

For a given total transmit power spectral density \(P_{\mathrm T}(f)\), the received signal power spectral density under maximum-ratio transmission (MRT) is given by
\begin{equation}\label{eq:PRS3}
P_{\mathrm{RS}}(f)=\|\mathbf{h}_{\mathrm{MISO}}(f)\|^2P_{\mathrm{T}}(f).
\end{equation}
Substituting (\ref{eq:PRS3}) into (\ref{eq:SNR}) gives the received SNR in terms of equivalent channel gain $\|\mathbf{h}_{\mathrm{MISO}}\|^2$ as
\begin{equation}\label{eq:SNR2}
    \mathrm{SNR}_{\mathrm{MISO}}(f)=\frac{\|\mathbf{h}_{\mathrm{MISO}}(f)\|^2P_{\mathrm{T}}(f)}{N_{0}}.
\end{equation}
Finally, the achievable rate over the
frequency band $[f_{\min},f_{\max}]$ is given by
\begin{equation}\label{eq:MISO_rate}
R_{\mathrm{MISO}}
=
\int_{f_{\min}}^{f_{\max}}
\log_2\left(
1+\mathrm{SNR}_{\mathrm{MISO}}(f)
\right)df .
\end{equation}

\subsection{Problem Formulation}\label{section_2_3}
In this work, we aim to characterize the achievable-rate upper bound for the considered MISO system imposed by the physical realizability constraints on the transmit-side MMN. To this end, we incorporate the multiport Bode--Fano constraints into the achievable-rate expression derived above and formulate the following optimization problem:
\begin{subequations}
\begin{align}
(\mathcal{P}_0):&
    \max_{\mathbf{S}_{\mathrm{M}}(f)} \;
    \int_{f_{\min}}^{f_{\max}}
    \log_2\left(
    1+\frac{\|\mathbf{h}_{\mathrm{MISO}}(f)\|^2P_{\mathrm{T}}(f)}{N_{0}}
    \right)\,df
    \label{eq:P0_obj}
    \\
    \mathrm{s.t.}\;
    &
    \int_0^\infty
    \xi_i(f)
    \log\left(\frac{1}{1-\mathcal{T}(f)}\right)
    df
    \leq
    B_{\mathrm{BF},i}, \quad \{i\}_1^{N_{\mathrm{BF}}},
    \label{eq:P0_BF}
    \\
    &
    0\leq \mathcal{T}(f)\leq 1.
    \label{eq:P0_box}
\end{align}
\end{subequations}

In problem $(\mathcal{P}_{0})$, the optimization variable is the scattering matrix $\mathbf S_{\mathrm M}(f)$ of the transmit-side MMN. The \(N_{\mathrm{BF}}\) inequalities in \eqref{eq:P0_BF} are the multiport Bode--Fano constraints, which characterize the fundamental tradeoff between power transmission efficiency and bandwidth by imposing weighted integral bounds on the logarithmic function of the transmission coefficient $\mathcal T(f)$. It indicates that \(\mathcal T(f)\) is considered as a finite matching resource. The transmission coefficient \(\mathcal{T}(f) \!\) is defined as the ratio of the expected net power accepted by the transmit antenna array to the expected incident power at the source-side ports of the MMN, i.e., % and characterizes the fundamental wideband limitation of physically realizable passive matching networks
\begin{equation}\label{eq:transmission_coef_definition}
    \mathcal{T}(f)
    =
    \frac{
    \mathbb{E}\left[
    \|\mathbf a_{\mathrm A}(f)\|^2
    -
    \|\mathbf b_{\mathrm A}(f)\|^2
    \right]
    }{
    \mathbb{E}\left[
    \|\mathbf a_{\mathrm T}(f)\|^2
    \right]
    } \in\![0,1].
\end{equation}
Under the isotropic source excitation assumption adopted in the multiport Bode-Fano constraint, $\mathcal{T}(f)$ can be expressed as
\begin{equation}\label{eq:transmission_coef_MISO}
\begin{aligned}
\mathcal{T}(f)
= \frac{1}{M}\operatorname{tr}\Big\{&
\mathbf{S}_{\mathrm{M},21}^{H}(f)
\big( \mathbf{I}_M-\mathbf{S}_{\mathrm{M},22}(f)\mathbf{S}_{\mathrm{T}}(f) \big)^{-H}\mathbf{W}(f) \\
&
\big( \mathbf{I}_M-\mathbf{S}_{\mathrm{M},22}(f)\mathbf{S}_{\mathrm{T}}(f) \big)^{-1} \mathbf{S}_{\mathrm{M},21}(f)
\Big\}.
\end{aligned}
\end{equation}
where $\mathbf{W}(f)=\mathbf{I}_M-\mathbf{S}_{\mathrm{T}}^H(f) \mathbf{S}_{\mathrm{T}}(f)$ denotes the power-acceptance matrix of the transmit antenna array. For the $i$-th constraint, the positive weighting function $\xi_i(f)$ and the bound $B_{\mathrm{BF},i}$ are determined by the transmit antenna array and can be computed following the procedure in~\cite{nie2016bandwidth2}.

Problem $(\mathcal P_0)$ is generally nonconvex and difficult to solve directly. The different blocks of the MMN scattering matrix are jointly coupled through matrix products and matrix inversion, with the resulting mappings further entering nonlinear composite functions in both the objective and constraints. Moreover, the optimization variable is a complex matrix-valued function over a continuous frequency range, and the integral Bode--Fano constraints couple the MMN responses across frequency. Consequently, problem \((\mathcal P_0)\) constitutes a challenging infinite-dimensional nonconvex functional optimization problem.

\subsection{Coupling Between Multiport Power Transmission and Propagation}
Beyond the challenges involved in solving problem \((\mathcal P_0)\), this formulation also exposes a fundamental issue in the communication-theoretic analysis of MMNs. The multiport Bode--Fano constraints are imposed on the transmission coefficient $\mathcal T(f)$, whereas the achievable-rate objective depends on the equivalent channel gain $\|\mathbf h_{\mathrm{MISO}}(f)\|^2$. The relationship between these two quantities is not explicit in problem \((\mathcal P_0)\). Clarifying this relationship is important for two reasons. First, it provides physical insight into how the power-transmission characteristics of an MMN translate into communication performance. Second, it may enable the matrix-valued optimization problem $(\mathcal P_0)$ to be reformulated in terms of the scalar transmission metric $\mathcal T(f)$, thereby substantially reducing its complexity. 

For transmit-side SMNs, previous studies have established a direct relationship between the SNR and the transmission coefficient \cite{deshpande2023achievable,deshpande2024generalization}. Specifically, the SNR can be expressed as
\begin{equation}
\mathrm{SNR}(f)
=
\mathrm{SNR}_{\mathrm{ideal}}(f)\mathcal T(f),
\end{equation}
where $\mathrm{SNR}_{\mathrm{ideal}}(f)$ is determined by the propagation channel and the remaining front-end characteristics excluding the SMN. Therefore, in the single-port case, the received SNR scales directly with the power transmission efficiency of the SMN. This direct proportionality makes the effect of the SMN on communication performance physically transparent and greatly simplifies the corresponding rate optimization. It stems from the scalar nature of the incident root power wave in an SMN, which leaves no directional degree of freedom. 

In contrast, the source-side incident root power wave of an MMN is vector-valued, introducing the excitation direction as an additional degree of freedom. For a given source-side excitation, the MMN shapes not only the total power delivered to the antenna array but also the complex amplitude and phase distribution across the antenna ports. The resulting antenna-port excitation is then weighted by the propagation channel in forming the received signal. 

Consequently, unlike the scalar factorization in the SMN case, the power-transmission behavior of the MMN and the propagation channel are inherently coupled in the equivalent channel gain, as reflected in problem \((\mathcal P_0)\). This coupling obscures the physical mechanism by which the MMN shapes the equivalent channel and prevents a direct reformulation of problem \((\mathcal P_0)\) in terms of the scalar transmission coefficient \(\mathcal T(f)\). Therefore, revealing this coupling is essential for both understanding the communication role of MMNs and establishing a tractable basis for simplifying problem $(\mathcal P_0)$.

\section{Interpretation of the Coupling Between Power Transmission and Propagation}\label{section_3}
To elucidate the coupling identified above, we first reformulate the equivalent MISO channel as a linear MIMO mapping induced by the MMN. Then, we utilize the directional gains of the resulting transfer matrix to establish an explicit relationship between the transmission coefficient and the equivalent channel gain.

\subsection{Transfer-Matrix Representation of the Equivalent Channel}\label{section_3_1}
Before analyzing the directional coupling between multiport power transmission and propagation, we rewrite the equivalent MISO channel in \eqref{eq:hMISO} as
\begin{equation} \label{eq:hMISO_system}
    \mathbf{h}_\mathrm{MISO}(f)= \mathbf{F}(f)\tilde{\mathbf{s}}_\mathrm{RT}(f),
\end{equation}
where 
\begin{align}
\mathbf{F}(f) &= \mathbf{S}_\mathrm{M,21}^H(f)\left(\mathbf{I}-\mathbf{S}_{\mathrm{M},22}(f)\mathbf{S}_{\mathrm{T}}(f)\right)^{-H}\mathbf{W}^{\frac{1}{2}}(f), \label{eq:F_matrix} \\
\tilde{\mathbf{s}}_\mathrm{RT}(f) &= \mathbf{W}^{-\frac{1}{2}}(f)\mathbf{s}_\mathrm{RT}^{*}(f). \label{eq:tilde_sRT}
\end{align}
The transmit antenna array is assumed to be strictly passive, ensuring that $\mathbf W(f)\succ\mathbf 0$ and that $\mathbf W^{1/2}(f)$ and $\mathbf W^{-1/2}(f)$ are well defined. Under the reformulation in (\ref{eq:hMISO_system}), $\mathbf{h}_\mathrm{MISO}(f)$ can be viewed as the output of a linear MIMO mapping with transfer matrix $\mathbf{F}(f)$ driven by the power-normalized propagation vector $\tilde{\mathbf{s}}_\mathrm{RT}(f)$. 

Compared with the original expression of $\mathbf{h}_\mathrm{MISO}(f)$ in (\ref{eq:hMISO}), the representation in (\ref{eq:hMISO_system}) provides two main advantages:
\begin{itemize}
    \item \textit{Passive directional gains.} The matrix $\mathbf F(f)$ is a contraction under the Euclidean norm. Hence,
    all singular values of $\mathbf F(f)$ are bounded by one and can be
    interpreted as physically admissible directional gains. This property reflects the passivity of the transmit-side MMN and the mutually coupled transmit antenna array.
    \item \textit{Connection to perfect matching.} The equality condition of the contraction property is directly related to complete source-side decoupling and matching. Specifically, it implies that the source-side input reflection matrix $\mathbf S_{\mathrm{in}}(f)$ becomes zero. In this case, $\|\mathbf{h}_{\mathrm{MISO}}(f)\|^2$ reaches its upper bound $\|\tilde{\mathbf{s}}_\mathrm{RT}(f)\|^2$, which corresponds to the equivalent channel gain under perfect impedance matching.
\end{itemize}

The properties supporting the above interpretation are summarized in the following proposition.
\begin{proposition}[\textit{Properties of the Transfer Matrix} $\mathbf{F}(f)$]\label{proposition1}
For each frequency \(f\), the transfer matrix \(\mathbf F(f)\) satisfies $\|\mathbf{F}(f)\mathbf{x}(f)\| \le \|\mathbf{x}(f)\|$ for all $\mathbf{x}(f)$, or equivalently $\mathbf 0 \preceq \mathbf F^H(f)\mathbf F(f) \preceq \mathbf I_M$. Moreover, the equality $\|\mathbf F(f)\mathbf x(f)\|=\|\mathbf x(f)\|$ for all $\mathbf x(f)$
holds if and only if the source-side input reflection matrix satisfies
$\mathbf S_{\mathrm{in}}(f)=\mathbf 0$.
\end{proposition}

\begin{IEEEproof}
Please refer to Appendix \ref{appendix_A}.
\end{IEEEproof}

\subsection{Directional-Gain Characterization}\label{section_3_2}
Applying the singular value decomposition to $\mathbf{F}(f)$ gives
\begin{equation} \label{eq:svd_F}
    \mathbf{F}(f)=\mathbf{U}(f)\mathbf{\Sigma}(f)\mathbf{V}^H(f),
\end{equation}
where \(\mathbf U(f)\) is unitary,
\(\mathbf V(f)=[\mathbf v_1(f),\ldots,\mathbf v_M(f)]\) is the unitary matrix of right singular vectors, and $\mathbf{\Sigma}(f)=\mathrm{diag}\left(\sigma_1\left( f\right), \dots,\sigma_M\left( f\right) \right)$. The singular values are ordered as \(\sigma_1(f)\geq\cdots\geq\sigma_M(f)\). In multivariable control theory, $\mathbf v_i(f)$ represents an input direction of $\mathbf F(f)$, and $\sigma_i(f)$ is the corresponding directional gain, which characterizes the ratio of the output magnitude to the input magnitude along this direction. By \textit{Proposition}~\ref{proposition1},
\(0\leq\sigma_i(f)\leq1\). Substituting (\ref{eq:svd_F}) into (\ref{eq:hMISO_system}) yields
\begin{equation}\label{eq:hMISO_direction_gain}
    \|\mathbf{h}_{\mathrm{MISO}}(f)\|^2=\sum_{i=1}^{M}\sigma_{i}^2(f)|z_{i}(f)|^2, 
\end{equation}
where $z_i(f)$ is the $i$-th entry of 
$\mathbf z(f)=\mathbf V^H(f)\tilde{\mathbf s}_{\mathrm{RT}}(f)$, and represents the projection coefficient of $\tilde{\mathbf s}_{\mathrm{RT}}(f)$ onto the $i$-th input direction of $\mathbf F(f)$. Similarly, the transmission coefficient in (\ref{eq:transmission_coef_MISO}) can be rewritten as
\begin{equation}\label{eq:transmission_direction_gain}
    \mathcal{T}(f)=\frac{\operatorname{tr}\left(\mathbf{F}\left(f\right)\mathbf{F}^H\left(f\right) \right)}{M}=\frac{\sum_{i=1}^{M}\sigma_{i}^2(f)}{M}.
\end{equation}

Equations~(\ref{eq:hMISO_direction_gain}) and (\ref{eq:transmission_direction_gain}) explicitly reveal the directional coupling between multiport power transmission and propagation. This relationship is illustrated in Fig.~\ref{directional_interpretation} using a three-dimensional example. The transmission coefficient $\mathcal T(f)$ is the average squared directional gain of $\mathbf F(f)$, whereas the equivalent channel gain is a propagation-weighted sum of the same squared directional gains, with the weights determined by the projections of $\tilde{\mathbf s}_{\mathrm{RT}}(f)$ onto the corresponding input directions. Therefore, the communication impact of an MMN depends not only on its overall power-transmission capability, but also on the alignment between the power-normalized propagation vector and high-gain input directions of \(\mathbf F(f)\). To quantify this alignment, we define the alignment coefficient as
\begin{equation}
\eta_{\gamma}(f)
=
\frac{\left\|
P_{\mathcal{S}_{\gamma}(f)}
\tilde{\mathbf{s}}_{\mathrm{RT}}(f)
\right\|^2}{\|\tilde{\mathbf{s}}_{\mathrm{RT}}(f)\|^2},
\end{equation}
where $P_{\mathcal S_\gamma(f)}$ denotes the orthogonal projector onto the dominant transmission subspace $\mathcal S_\gamma(f)=\operatorname{span}\{\mathbf v_i(f):\sigma_i^2(f)\geq\gamma\sigma_1^2(f)\}$, and $\gamma\in(0,1]$ is a prescribed threshold. A larger $\eta_{\gamma}(f)$ indicates better alignment of the power-normalized propagation vector with the dominant transmission subspace.

\begin{figure}
\centering
\includegraphics[scale=1.2]{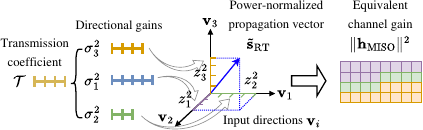}
\caption{Illustration of the directional-gain relationship between the transmission coefficient and the equivalent channel gain in a three-dimensional example.}\label{directional_interpretation} %图片标题
\end{figure}

\begin{remark}[\textit{Insight into Enhancing the Equivalent Channel Gain}]\label{remark1}
The directional-gain characterization shows that the transmission coefficient alone does not uniquely determine the equivalent channel gain. For a given \(\mathcal T(f)\), allocating larger directional gains to input directions that are well aligned with \(\tilde{\mathbf s}_{\mathrm{RT}}(f)\) yields a larger equivalent channel gain than allocating them to weakly aligned directions. Therefore, the equivalent channel gain can be enhanced by appropriately shaping the directional gains of $\mathbf F(f)$ according to the power-normalized propagation vector. This insight provides a basis for reducing the complexity of problem $(\mathcal P_0)$.
\end{remark}

\section{Achievable-Rate Upper Bound Under Multiport Bode--Fano Constraints}\label{section_4}
Motivated by the insight in \textit{Remark}~\ref{remark1}, we derive a tight upper bound on the equivalent channel gain and reformulate the problem $(\mathcal{P}_{0})$ as a scalar transmission-allocation problem. We then characterize the optimal transmission-coefficient structure and develop a greedy constraint-wise repair method to obtain a feasible suboptimal solution.
\subsection{Scalar Reformulation of the Rate Upper-Bound Problem}\label{section_4_1}
Using the transfer-matrix representation introduced in Section~\ref{section_3_1}, problem \((\mathcal P_0)\) can be expressed in terms of \(\mathbf F(f)\) as
\begin{subequations}
\begin{align}
(\mathcal{P}_1):
    &\max_{\mathbf{F}(f)} \;
    \int_{f_{\min}}^{f_{\max}}
    \log_2\left(
    1+\frac{\|\mathbf{F}(f)\tilde{\mathbf{s}}_\mathrm{RT}(f)\|^2P_{\mathrm{T}}(f)}{N_{0}}
    \right)\,df
    \label{eq:P2_obj}
    \\
    \mathrm{s.t.}\;
    &
    \int_0^\infty
    \xi_i(f)
    \log\left(\frac{1}{1-\mathcal{T}(f)}\right)
    df
    \leq
    B_{\mathrm{BF},i},\quad \{i\}_1^{N_{\mathrm{BF}}},
    \label{eq:P2_cons1}
    \\
    &
    \mathcal{T}(f)=\frac{\operatorname{tr}\left(\mathbf{F}\left(f\right)\mathbf{F}^H\left(f\right) \right)}{M},
    \label{eq:P2_cons2}
    \\
    & 
    \mathbf{0} \preceq \mathbf{F}\left(f\right)\mathbf{F}^{H}\left(f\right) \preceq  \mathbf{I}.
    \label{eq:P2_cons3}
\end{align}
\end{subequations}
Although problem \((\mathcal P_1)\) still involves the matrix-valued function \(\mathbf F(f)\), a tight upper bound on the equivalent channel gain can be established in terms of the transmission coefficient \(\mathcal T(f)\), as stated in the following lemma.
\begin{lemma}[\textit{Tight upper bound on the equivalent channel gain}]\label{lemma1}
Let \(\mathbf A \in \mathbb{C}^{M\times M}\) satisfy
\(\mathbf{0} \preceq \mathbf{A}\mathbf{A}^{H} \preceq \mathbf{I}_{M}\), and define $\tau=\frac{\operatorname{tr}\left(\mathbf{A}\mathbf{A}^H\right)}{M}$. Then, for any nonzero vector \(\mathbf u\in\mathbb C^{M\times1}\),
\begin{equation}\label{eq:upper_bound}
\left\|\mathbf A\mathbf{u}\right\|^2
\leq
\left\|\mathbf{u}\right\|^2
\min\{M \tau,1\}.
\end{equation}
The bound is tight for every \(\tau\in[0,1]\). In particular, when
\(0\le \tau\le \frac{1}{M}\), equality is attained by any matrix of the form
\begin{equation}\label{eq:equality_case1}
\mathbf A
=
\sqrt{M\tau}\,
\mathbf q
\frac{\mathbf u^H}{\|\mathbf u\|},
\end{equation}
where \(\mathbf q\in\mathbb C^M\) is an arbitrary unit-norm vector. This rank-one construction assigns the entire sum \(M\tau\) of squared directional gains to the input direction \(\mathbf u\). For \(\frac{1}{M}<\tau\le 1\), an equality-achieving construction is provided in Appendix~\ref{appendix_B}.
\end{lemma}
\begin{IEEEproof}
Please refer to Appendix \ref{appendix_B}.
\end{IEEEproof}

Since \(\|\tilde{\mathbf{s}}_{\mathrm{RT}}(f)\|^2\) is the equivalent channel gain under perfect impedance matching, we introduce the corresponding ideal-matching SNR to simplify the subsequent expressions:
\begin{equation}
\mathrm{SNR}_{\mathrm{ideal}}^{\mathrm{MISO}}(f)=\|\tilde{\mathbf{s}}_{\mathrm{RT}}(f)\|^2 \frac{P_{\mathrm{T}}(f)}{N_{0}}
\end{equation}
For the same purpose, define $\mathcal{T}_{1}(f)=\min\{M \mathcal{T}(f),1\}$. According to \textit{Lemma} \ref{lemma1}, for any feasible $\mathbf F(f)$ in problem $(\mathcal P_1)$, the maximum attainable value of $\|\mathbf{F}(f)\tilde{\mathbf{s}}_\mathrm{RT}(f)\|^2$ is $\|\tilde{\mathbf{s}}_{\mathrm{RT}}(f)\|^2\mathcal{T}_{1}(f)$. Since this bound is tight and the Bode--Fano constraints depend on \(\mathbf F(f)\) only through the induced transmission coefficient \(\mathcal T(f)\), the matrix-valued degrees of freedom in \(\mathbf F(f)\) can be eliminated. Therefore, problem
\((\mathcal P_1)\) is equivalently reformulated as the following scalar functional optimization problem over \(\mathcal T(f)\):
\begin{subequations}
\begin{align}
(\mathcal{P}_2):
    &\max_{\mathcal{T}(f)} \;\int_{f_{\min}}^{f_{\max}}\log_2\left(1+\mathrm{SNR}_{\mathrm{ideal}}^{\mathrm{MISO}}(f)\mathcal{T}_{1}(f)
    \right)\,df
    \label{eq:P3_obj}
    \\
    \mathrm{s.t.}\;
    &
    \int_0^\infty
    \xi_i(f)
    \log\left(\frac{1}{1-\mathcal{T}(f)}\right)
    df
    \leq
    B_{\mathrm{BF},i},\quad \{i\}_1^{N_{\mathrm{BF}}},
    \label{eq:P3_cons1}
    \\
    & 
    0 \leq \mathcal{T}(f) \leq 1
    \label{eq:P3_cons2}
\end{align}
\end{subequations}

In problem $(\mathcal P_2)$, increasing $\mathcal T(f)$ beyond $\frac{1}{M}$ does not improve the achievable-rate objective but only tightens the Bode--Fano constraints. Hence, without loss of optimality, $\mathcal T(f)$ can be restricted to $0\leq\mathcal T(f)\leq\frac{1}{M}$. Over this range, $\mathcal T_1(f)=M\mathcal T(f)$, and problem $(\mathcal P_2)$ reduces to
the following scalar transmission-allocation problem:
\begin{subequations}
\begin{align}
(\mathcal{P}_3):
    &\max_{\mathcal{T}(f)} \;\int_{f_{\min}}^{f_{\max}}\log_2\left(1+\mathrm{SNR}_{\mathrm{ideal}}^{\mathrm{MISO}}(f) M\mathcal{T}(f)
    \right)\,df
    \label{eq:P4_obj}
    \\
    \mathrm{s.t.}\;
    &
    \int_0^\infty
    \xi_i(f)
    \log\left(\frac{1}{1-\mathcal{T}(f)}\right)
    df
    \leq
    B_{\mathrm{BF},i},\quad \{i\}_1^{N_{\mathrm{BF}}},
    \label{eq:P4_cons1}
    \\
    & 
    0 \leq \mathcal{T}(f) \leq \frac{1}{M}
    \label{eq:P4_cons2}
\end{align}
\end{subequations}

Therefore, problems \((\mathcal P_0)\) and \((\mathcal P_3)\) have the same optimal value, and the achievable-rate upper bound characterized by
\((\mathcal P_0)\) can be obtained by solving problem \((\mathcal P_3)\). Given an optimal scalar solution \(\mathcal T^\star(f)\) of problem \((\mathcal P_3)\), a corresponding optimal transfer matrix \(\mathbf F^\star(f)\) of \((\mathcal P_1)\) can be constructed using \eqref{eq:equality_case1}. In this construction,
\(\tilde{\mathbf s}_{\mathrm{RT}}(f)\) is an input direction of \(\mathbf F^\star(f)\) with directional gain \(\sqrt{M\mathcal T^\star(f)}\). Accordingly, attaining the upper bound
requires the MMN scattering matrix, together with the antenna-array scattering matrix, to induce the corresponding transfer matrix \(\mathbf F^\star(f)\).
\begin{remark}[\textit{Insights for Communication-Oriented MMN Design}]\label{remark2}
For the considered MISO system, the ideal-matching SNR can be attained with only $\mathcal T(f)=\frac{1}{M}$ when $\mathbf F(f)$ provides unit directional gain along the power-normalized propagation vector. Hence, perfect impedance matching, corresponding to $\mathcal T(f)=1$, is not necessary for maximizing communication performance. Under this condition, the propagation vector lies entirely in the dominant transmission subspace, yielding $\eta_{\gamma}(f)\!=\!1$ for any $\gamma \! \in \!(0,1]$. Since the Bode--Fano constraints limit the available matching capability across frequency, a communication-oriented MMN should prioritize aligning its dominant
transmission subspace with the power-normalized propagation vector rather than merely maximizing the transmission coefficient.
\end{remark}

\subsection{Optimal Transmission Coefficient Allocation}
Since problem \((\mathcal P_3)\) is convex, the KKT conditions are necessary and sufficient for optimality. The Lagrangian can be constructed as
\begin{equation}\label{eq:lagrangian}
\begin{aligned}
\mathcal{L}&\left(
\mathcal{T}(f), \{\mu_i\}_{i=1}^{N_{\mathrm{BF}}+2}
\right)
\\
&=-
\int_{f_{\min}}^{f_{\max}}
\log_2\left(
1+\mathrm{SNR}_{\mathrm{ideal}}^{\mathrm{MISO}}(f)
M\mathcal{T}(f)
\right)\,df
\\
&\quad+
\sum_{i=1}^{N_{\mathrm{BF}}}\mu_i
\left[
\int_{0}^{\infty}
\xi_i(f)
\log\left(
\frac{1}{1-\mathcal{T}(f)}
\right)\,df
-
B_{\mathrm{BF},i}
\right]
\\
&\quad-
\mu_{N_{\mathrm{BF}}+1}\mathcal{T}(f)
+
\mu_{N_{\mathrm{BF}}+2}
\left(\mathcal{T}(f)-\frac{1}{M}\right),
\end{aligned}
\end{equation}
where $\mu_{i} \geq 0$ is the Lagrange multiplier associated with the $i$-th inequality constraint. The corresponding complementary slackness conditions are
\begin{subequations}\label{eq:complementary_slackness}
\begin{align}
\mu_i^\star
\left[
\int_{0}^{\infty}
\xi_i(f)
\log\left(\frac{1}{1-\mathcal{T}^\star(f)}\right)
\,df
-
B_{\mathrm{BF},i}
\right]
=0,
\, \{i\}_1^{N_{\mathrm{BF}}},
\label{eq:cs_bf}
\\
\mu_{N_{\mathrm{BF}}+1}^\star\mathcal{T}^\star(f)
=0, \quad \mu_{N_{\mathrm{BF}}+2}^\star
\left(\mathcal{T}^\star(f)-\frac{1}{M}\right)
=0.
\end{align}
\end{subequations}
To derive the closed-form expression of the interior solution, we first set \(\mu_{N_{\mathrm{BF}}+1}^\star=\mu_{N_{\mathrm{BF}}+2}^\star=0\) and solve the resulting stationarity condition. The pointwise bounds are subsequently incorporated by clipping the resulting expression to the interval \([0,1/M]\). 

Let \(\delta(f)\) be an arbitrary admissible variation, and let \(\epsilon\) denote its scalar amplitude. The first-order stationarity condition from variational calculus is
\begin{equation}\label{eq:stability_cond}
\left.\frac{\mathrm{d}}{\mathrm{d}\epsilon}\left[\mathcal{L}\left(\mathcal{T}^\star(f)+\epsilon\delta(f)\right)\right]\right|_{\epsilon=0} = 0
\end{equation}
Substituting \eqref{eq:lagrangian} into (\ref{eq:stability_cond}) and collecting the terms associated with \(\delta(f)\), we obtain the first-variation condition in \eqref{eq:simplified_sta_cond}.
\begin{figure*}[hb]
\centering 
\hrulefill
\vspace*{2pt} 
\begin{equation}\label{eq:simplified_sta_cond}
\int_{0}^{\infty}
\frac{\delta(f)}{\ln 2}\left[
\left(
\frac{-M\cdot\mathrm{SNR}_{\mathrm{ideal}}^{\mathrm{MISO}}(f)}{1+M\cdot\mathrm{SNR}_{\mathrm{ideal}}^{\mathrm{MISO}}(f) \mathcal{T}^{\star}(f)}
\right)
+
\sum_{i=1}^{N_{\mathrm{BF}}}
\mu_i^\star
\frac{\ln2\xi_{i}(f)}{1-\mathcal{T}^{\star}(f)}
\right]
\mathrm{d}f-\mu_{N_{\mathrm{BF}+1}}^\star \delta(f)
+\mu_{N_{\mathrm{BF}+2}}^\star \delta(f)=0
\end{equation}
\end{figure*}
With $\mu_{N_{\mathrm{BF}}+1}^\star=0$ and
$\mu_{N_{\mathrm{BF}}+2}^\star=0$, the last two terms in
(\ref{eq:simplified_sta_cond}) vanish. Since $\delta(f)$ is arbitrary, the remaining coefficient of \(\delta(f)\) must be zero, which gives
\begin{equation} \label{eq:station_simplied1}
\left(
\frac{-M\cdot\mathrm{SNR}_{\mathrm{ideal}}^{\mathrm{MISO}}(f)}{1+M\cdot\mathrm{SNR}_{\mathrm{ideal}}^{\mathrm{MISO}}(f) \mathcal{T}^{\star}(f)}
\right)
+
\sum_{i=1}^{N_{\mathrm{BF}}}
\mu_i^\star
\frac{\ln2\xi_{i}(f)}{1-\mathcal{T}^{\star}(f)} = 0
\end{equation}
Let $\boldsymbol{\mu}^\star=[\mu_1^\star,\ldots,\mu_{N_{\mathrm{BF}}}^\star]^T$ collect the optimal Lagrange multipliers associated with the Bode--Fano constraints. Solving \eqref{eq:station_simplied1} for \(\mathcal T^\star(f)\) and enforcing the pointwise bounds \(0\leq \mathcal T(f)\leq 1/M\) yield its closed-form expression. Together with the complementary slackness condition in \eqref{eq:cs_bf}, primal feasibility, and dual feasibility, the optimal transmission coefficient and the corresponding
Lagrange multipliers are characterized by
\begin{subequations}\label{eq:opt_mu_T_eqs}
\begin{align}
&\mathcal{T}^\star(f;\boldsymbol{\mu}^\star) = \left[
\frac{1 - \ln 2 \frac{\sum_{i=1}^{N_{\mathrm{BF}}}\mu_i^\star \xi_{i}(f)}{M\cdot\mathrm{SNR}_{\mathrm{ideal}}^{\mathrm{MISO}}(f)}}
{1 + \ln 2 \sum_{i=1}^{N_{\mathrm{BF}}} \mu_i^\star \xi_{i}(f)}
\right]_{0}^{\frac{1}{M}},\quad \boldsymbol{\mu}^\star \ge 0 
\label{eq:station_simplied2}
\\
&\mu_i^\star \left[\int_{0}^{\infty} \xi_i(f) \log\left(\frac{1}{1-\mathcal{T}^\star(f)}\right) \,df-B_{\mathrm{BF},i} \right]=0, \, \{i\}_{1}^{N_{\mathrm{BF}}} 
\label{eq:cs_bf2}
\\
&\int_{0}^{\infty} \xi_i(f) \log\left(\frac{1}{1-\mathcal{T}^\star(f)}\right) \,df \leq B_{\mathrm{BF},i},\, \{i\}_{1}^{N_{\mathrm{BF}}} 
\label{eq:pri_fea_bf2}
\end{align}
\end{subequations}

\begin{remark}[\textit{Insights into Optimal Transmission Allocation}]\label{remark3}
The closed-form solution in \eqref{eq:station_simplied2} shows that the optimal transmission coefficient should be allocated nonuniformly across frequency according to the tradeoff between achievable-rate benefit and Bode--Fano matching cost. Specifically, frequencies with a higher ideal-matching $\mathrm{SNR}_{\mathrm{ideal}}^{\mathrm{MISO}}(f)$ and a smaller effective matching cost $\sum_{i=1}^{N_{\mathrm{BF}}}\mu_i^\star\xi_i(f)$ are assigned larger transmission coefficients.
\end{remark}

To facilitate the solution of the KKT system in \eqref{eq:opt_mu_T_eqs}, we introduce the following notation. For any nonnegative multiplier vector \(\boldsymbol{\mu}=[\mu_1,\ldots,\mu_{N_{\mathrm{BF}}}]^T\), let \(\mathcal T_{\boldsymbol{\mu}}(f)\) denote the transmission coefficient obtained from \eqref{eq:station_simplied2} by replacing \(\boldsymbol{\mu}^\star\) with \(\boldsymbol{\mu}\). The corresponding
cost of the \(i\)-th Bode--Fano constraint is defined as
\begin{equation}\label{eq:C_mu_def}
C_i(\boldsymbol{\mu})
\triangleq
\int_{0}^{\infty}
\xi_i(f)
\log\left(
\frac{1}{1-\mathcal T_{\boldsymbol{\mu}}(f)}
\right)
\,df ,
\quad \{i\}_{1}^{N_{\mathrm{BF}}}.
\end{equation}
The KKT system in \eqref{eq:opt_mu_T_eqs} exhibits a monotonic structure. Since the Bode--Fano weighting functions $\xi_{i}(f)$ are nonnegative, increasing any component of \(\boldsymbol{\mu}\) makes $\mathcal{T}_{\boldsymbol{\mu}}(f)$ nonincreasing at any frequency. Consequently, each $C_{i}(\boldsymbol{\mu})$ is componentwise nonincreasing with respect to \(\boldsymbol{\mu}\). Thus, increasing the multipliers facilitates satisfaction of the Bode--Fano constraints, but at the cost of lower transmission coefficients and a reduced achievable rate. This tradeoff motivates increasing each multiplier only as much as needed to restore feasibility.

Based on the above monotonicity, a greedy constraint-wise repair method to obtain a feasible approximate solution to (\ref{eq:opt_mu_T_eqs}) is proposed. Let \(\boldsymbol{\mu}^{(t)}=[\mu_1^{(t)},\ldots,\mu_{N_{\mathrm{BF}}}^{(t)}]^T\) denote the multiplier vector at iteration \(t\), and initialize the search with \(\boldsymbol{\mu}^{(0)}=\mathbf 0\). The corresponding transmission coefficient satisfies $\mathcal{T}_{\boldsymbol{\mu}^{(0)}}(f)=\frac{1}{M}$ at each frequency, thereby maximizing the achievable-rate objective when the Bode--Fano
constraints are ignored. This initial allocation also maximizes each Bode--Fano cost \(C_i(\boldsymbol{\mu}^{(0)})\) and may therefore violate one or more Bode--Fano constraints. At each iteration \(t\), the relative violations of all Bode--Fano constraints are evaluated, and the \(k\)-th constraint with the largest relative violation is selected. With all other multipliers fixed, a bisection search is performed over \(\mu_k\) until the \(k\)-th Bode--Fano constraint is satisfied from the feasible side, that
is, \(C_k(\boldsymbol{\mu^{(t)}})\leq B_{\mathrm{BF},k}\), with the relative
residual below a prescribed tolerance \(\varepsilon\). This update repairs the selected constraint while avoiding an excessive increase of \(\mu_k\), thereby limiting the associated rate loss. The procedure is repeated until all Bode--Fano constraints are satisfied. The detailed implementation is summarized in Algorithm~\ref{alg:greedy_constraint_repair}.

\begin{algorithm}[t]
\caption{Greedy Constraint-Wise Repair Method}
\label{alg:greedy_constraint_repair}
\DontPrintSemicolon

\KwIn{
Bode--Fano parameters \(\xi_{i}(f)\) and \(B_{\mathrm{BF},i}\), tolerance \(\varepsilon\).
}

\KwOut{
Suboptimal transmission coefficient \(\mathcal T_{\widehat{\boldsymbol{\mu}}}(f)\) and multiplier vector \(\widehat{\boldsymbol{\mu}}\).
}

\textbf{Initialize:} \(\boldsymbol{\mu}^{(0)}=\mathbf 0\) and iteration index $t=0$.\;

Evaluate \(C_i(\boldsymbol{\mu}^{(t)})\), \(i=1,\ldots,N_{\mathrm{BF}}\).\;

\While{
\(\displaystyle\max_{i}\frac{C_i(\boldsymbol{\mu}^{(t)})-B_{\mathrm{BF},i}}{B_{\mathrm{BF},i}} >
0\)
}
{
    Select the most violated constraint:
    \[
    k
    =
    \arg\max_i
    \left[
    \frac{C_i(\boldsymbol{\mu}^{(t)})-B_{\mathrm{BF},i}}
    {B_{\mathrm{BF},i}}
    \right]_+ .
    \]

    Set $\boldsymbol{\mu}_{\mathrm{lo}}=\boldsymbol{\mu}^{(t)}$, \(\boldsymbol{\mu}_{\mathrm{hi}}=\boldsymbol{\mu}^{(t)}\), and $\mu_{\mathrm{hi}}=1$.
    
    Set \(\boldsymbol{\mu}_{\mathrm{hi}}\left[k\right]=\mu_{\mathrm{hi}}\).\;

    \While{\(C_k(\boldsymbol{\mu}_{\mathrm{hi}})>B_{\mathrm{BF},k}\)}
    {
        Set \(\mu_{\mathrm{hi}}=2\mu_{\mathrm{hi}}\) and \(\boldsymbol{\mu}_{\mathrm{hi}}\left[k\right]=\mu_{\mathrm{hi}}\) .\;
    }
    Set \(\boldsymbol{\mu}_{\mathrm{mid}}=\frac{\boldsymbol{\mu}_{\mathrm{hi}}+\boldsymbol{\mu}_{\mathrm{lo}}}{2}\).\;
    \While{
    \(\displaystyle
    \frac{
    \left|C_k(\boldsymbol{\mu}_{\mathrm{mid}}) \!-\!B_{\mathrm{BF},k}\right|
    }{
    B_{\mathrm{BF},k}
    }\!
    >\!
    \varepsilon
    \) or $C_k(\boldsymbol{\mu}_{\mathrm{mid}})\!>\!B_{\mathrm{BF},k}$
    }
    {
        \eIf{\(C_k(\boldsymbol{\mu}_{\mathrm{mid}})>B_{\mathrm{BF},k}\)}
        {
            Set \(\boldsymbol{\mu}_{\mathrm{lo}}=\boldsymbol{\mu}_{\mathrm{mid}}\).\;
        }
        {
            Set \(\boldsymbol{\mu}_{\mathrm{hi}}=\boldsymbol{\mu}_{\mathrm{mid}}\).\;
        }
        Set
   $\boldsymbol{\mu}_{\mathrm{mid}}=\frac{\boldsymbol{\mu}_{\mathrm{hi}}+\boldsymbol{\mu}_{\mathrm{lo}}}{2}$.
       
    }

    Set $t=t+1$ \, \text{and} \, \(\boldsymbol{\mu}^{(t)}=\boldsymbol{\mu}_{\mathrm{mid}}\).\;

    Evaluate \(C_i(\boldsymbol{\mu}^{(t)})\), \(i=1,\ldots,N_{\mathrm{BF}}\).\;
}
Set
\(\widehat{\boldsymbol{\mu}}
=\boldsymbol{\mu}^{(t)}\).\;
\Return \(\mathcal T_{\widehat{\boldsymbol{\mu}}}(f)\) and \(\widehat{\boldsymbol{\mu}}\).\;

\end{algorithm}
It should be noted that the proposed method does not solve the full KKT system exactly. Increasing a multiplier selected at a later iteration may further
reduce \(\mathcal T_{\boldsymbol{\mu}}(f)\) over the frequency band. Consequently, a constraint that was previously made approximately active from the feasible side may become strictly feasible after subsequent repairs. The final multiplier vector therefore may not satisfy all complementary slackness conditions, and the resulting feasible solution may be conservative.

\section{Circuit Realization of Multiport Matching Networks}\label{section_5}
In this section, we first establish the mapping from the parameters of a multiport matching circuit to the achievable rate and formulate a rate-maximization problem for a prescribed circuit topology. Then, a multi-start gradient-based optimization method with structured initializations is developed to optimize the circuit parameters. 

\subsection{From Circuit Parameters to Achievable Rate}\label{section_5_1}
For a prescribed multiport circuit topology, we directly optimize the circuit parameters by maximizing the achievable rate, rather than fitting the theoretically optimal transmission coefficient \(\mathcal T^\star(f)\) or an equality-achieving transfer-matrix response \(\mathbf F^\star(f)\). This rate-oriented formulation is adopted because a smaller fitting error in \(\mathcal T^\star(f)\) or \(\mathbf F^\star(f)\) does not necessarily translate into a higher achievable rate. As indicated by (\ref{eq:P4_obj}) and (\ref{eq:P2_obj}), fitting errors at different frequencies affect the final rate unequally, since the rate integrand depends on the frequency-dependent \(\mathrm{SNR}_{\mathrm{ideal}}^{\mathrm{MISO}}(f)\) or the propagation vector \(\tilde{\mathbf{s}}_\mathrm{RT}(f)\), and is further shaped by the logarithmic rate function.

\begin{figure}
\centering
\includegraphics[scale=0.4]{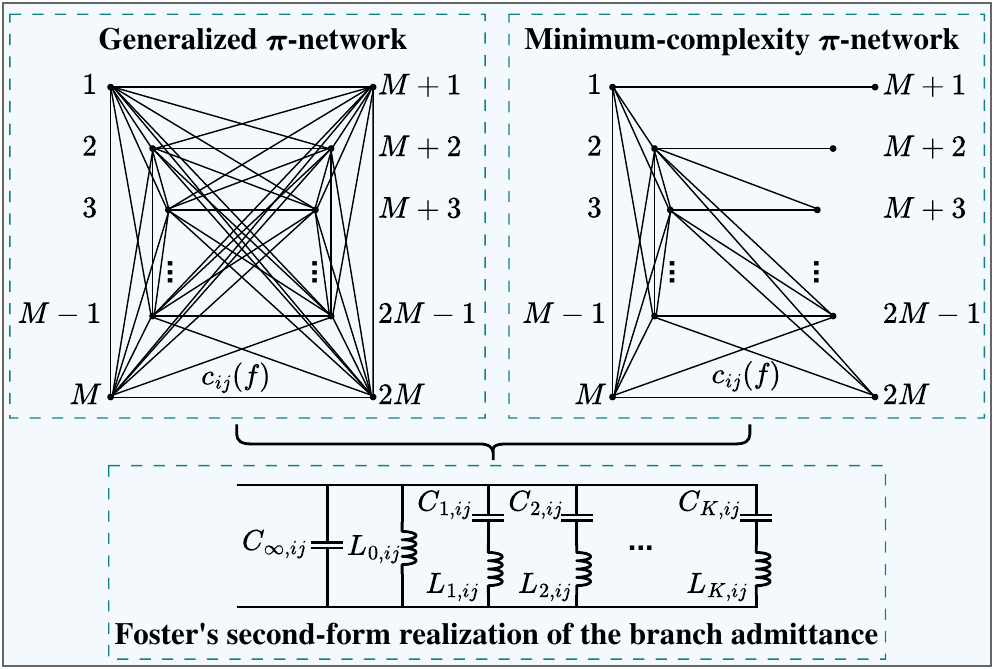}
\caption{Generalized and minimum-complexity $\Pi$-network topologies and the realization of each branch admittance $c_{ij}(f)$ in Foster's second form. Grounded shunt branches are omitted for clarity.}\label{circuit_structure} %图片标题
\end{figure}

As illustrated in Fig.~\ref{circuit_structure}, a generalized \(\Pi\)-network and a minimum-complexity \(\Pi\)-network are considered as two prescribed multiport circuit topologies. The \(2M\) nodes are indexed such that nodes \(1,\ldots,M\) correspond to the source-side ports and nodes \(M+1,\ldots,2M\) correspond to the antenna-side ports. Each topology is specified by a branch set $\mathcal B\subseteq\{(i,j)\mid 1\le i\le j\le 2M\}$, where \(i<j\) denotes an inter-node branch and \(i=j\) denotes a shunt branch connected to the common ground. The generalized \(\Pi\)-network contains all admissible branches, providing the largest number of circuit degrees of freedom for shaping the multiport transfer response and aligning the power-normalized propagation vector with the dominant transmission subspace. By contrast, the minimum-complexity \(\Pi\)-network retains only the minimum number of branches required to realize single-frequency multiport conjugate matching~\cite{nie2014systematic}. For each \((i,j)\in\mathcal B\), \(c_{ij}(f)\) denotes the corresponding frequency-dependent branch admittance.

Following classical one-port LC network synthesis, each connected branch admittance \(c_{ij}(f)\) is parameterized by a finite Foster's second form \cite{rodriguez2021systematic} with \(K\) series-LC resonant sections. This structure realizes a passive and lossless one-port admittance using positive inductors and capacitors. As shown in Fig. \ref{circuit_structure}, each branch realization consists of a shunt capacitor \(C_{\infty,ij}\), a shunt inductor \(L_{0,ij}\), and \(K\) parallel resonant sections, where the \(m\)-th section comprises \(L_{m,ij}\) and \(C_{m,ij}\) connected in series. Increasing \(K\) provides greater flexibility in shaping the branch admittance over the target frequency band, at the cost of increasing the dimension of the circuit-parameter optimization problem.

With the topology and branch realization specified, we next establish the mapping from the circuit parameters to the achievable rate. Since capacitance and inductance values in matching circuits are typically on the order of pF and nH, respectively, and may span several orders of magnitude, their logarithms are used as the optimization variables to improve numerical scaling. For branch \((i,j)\), the optimization variable vector is defined as
\begin{equation*}
\boldsymbol{\theta}_{ij}=
\left[
\theta_{ij,\infty}^{C},
\theta_{ij,0}^{L},
\theta_{ij,1}^{L},
\theta_{ij,1}^{C},
\ldots,
\theta_{ij,K}^{L},
\theta_{ij,K}^{C}
\right]^T
\in\mathbb R^{2K+2},
\end{equation*}
where the physical component values are
\begin{equation*}
\begin{aligned}
&C_{\infty,ij}=\exp\left(\theta_{ij,\infty}^{C}\right),\quad
L_{0,ij}=\exp\left(\theta_{ij,0}^{L}\right), \\
&L_{m,ij}=\exp\left(\theta_{ij,m}^{L}\right),\quad
C_{m,ij}=\exp\left(\theta_{ij,m}^{C}\right),\, m=1,\ldots,K.
\end{aligned}
\end{equation*}
The resulting Foster second-form branch admittance is
\begin{equation}\label{eq:foster_branch_admittance}
\begin{aligned}
c_{ij}(f;\boldsymbol{\theta}_{ij})
=&
\mathrm{j}2\pi f C_{\infty,ij}
+\frac{1}{\mathrm{j}2\pi f L_{0,ij}} \\
&+
\sum_{m=1}^{K}
\frac{1}{
\mathrm{j}2\pi f L_{m,ij}
+
\frac{1}{\mathrm{j}2\pi f C_{m,ij}}
}.
\end{aligned}
\end{equation}
The network admittance matrix is assembled from the branch admittances using the standard branch-admittance stamping rule \cite{hao2022realizable}
\begin{equation}\label{eq:Y_stamping}
\mathbf{Y}(f;\boldsymbol{\theta})
=
\sum_{(i,j)\in\mathcal{B}}
c_{ij}(f;\boldsymbol{\theta}_{ij})\mathbf{E}_{ij},
\end{equation}
where \(\boldsymbol{\theta}
=\{\boldsymbol{\theta}_{ij}:(i,j)\in\mathcal B\}\) collects all optimization variables, and $\mathbf{E}_{ij}\in\mathbb{R}^{2M\times 2M}$ is the constant stamp matrix of branch $(i,j)$. For an inter-node branch $i<j$,
\begin{equation}\label{eq:Eij_inter_node}
\begin{aligned}
{}[\mathbf{E}_{ij}]_{ii}
=
[\mathbf{E}_{ij}]_{jj}
&=1,\\
[\mathbf{E}_{ij}]_{ij}
=
[\mathbf{E}_{ij}]_{ji}
&=-1,
\end{aligned}
\end{equation}
with all other entries being zero. For a grounded branch $i=j$, the stamp matrix reduces to
\begin{equation}\label{eq:Eij_shunt}
[\mathbf{E}_{ii}]_{ii}=1,
\end{equation}
with all other entries being zero. Subsequently, the scattering matrix of the MMN is obtained from the admittance matrix as
\begin{equation}\label{eq:Y_to_S}
\mathbf S_{\mathrm M}(f;\boldsymbol{\theta})=
\left(
\mathbf I_{2M}-
Z_0\mathbf Y(f;\boldsymbol{\theta})
\right)
\left(
\mathbf I_{2M}
+
Z_0\mathbf Y(f;\boldsymbol{\theta})
\right)^{-1}.
\end{equation}
Finally, substituting the resulting $\mathbf S_{\mathrm M}(f;\boldsymbol{\theta})$ into (\ref{eq:hMISO}) and (\ref{eq:MISO_rate}) yields the equivalent channel $\mathbf h_{\mathrm{MISO}}(f;\boldsymbol{\theta})$ and the achievable rate $R_{\mathrm{MISO}}(\boldsymbol{\theta})$, respectively.

In practical implementations, the component values are restricted to finite realizable ranges, denoted by \([C_{\min},C_{\max}]\) and \([L_{\min},L_{\max}]\). Accordingly, for a prescribed topology \(\mathcal B\) and a fixed number of resonant sections \(K\), the rate-oriented circuit realization problem is formulated as
\begin{subequations}\label{circuit_realize}
\begin{align}
(\mathcal{P}_5): &\max_{\boldsymbol{\theta}} \;
R_{\mathrm{MISO}}(\boldsymbol{\theta})
\label{circuit_target}
\\
\mathrm{s.t.}\;
&
\ln C_{\min}
\leq
\theta_{ij,\kappa}^{C}
\leq
\ln C_{\max},
\;
\kappa\in\mathcal K_C,\;
(i,j)\in\mathcal B, \label{circuit_cs1}
\\
&
\ln L_{\min}
\leq
\theta_{ij,\ell}^{L}
\leq
\ln L_{\max},
\;
\ell\in\mathcal K_L,\;
(i,j)\in\mathcal B. \label{circuit_cs2}
\end{align}
\end{subequations}
where $\mathcal K_C=\{\infty,1,\ldots,K\}$ and $\mathcal K_L=\{0,1,\ldots,K\}$ denote the index sets of the capacitance and inductance variables, respectively.

\subsection{Multi-Start Gradient-Based Optimization}\label{section_5_2}
Since problem \((\mathcal P_{5})\) is nonconvex, a gradient-based solver may converge to different stationary points depending on the initialization. To mitigate this issue, we solve problem \((\mathcal P_{5})\) from multiple structured initializations and retain the solution that yields the highest achievable rate. Rather than randomly generating capacitance and inductance values, the starting points are constructed by combining physically meaningful component scales with prescribed resonant-frequency configurations. This strategy avoids arbitrary combinations of inductance and capacitance values that may yield severely mismatched initial responses, thereby improving the effectiveness of the multi-start optimization.

The structured initializations are constructed as follows. For the target frequency band \([f_{\min},f_{\max}]\), the center frequency is defined as
\begin{equation}
f_{\mathrm{c}}
=
\frac{f_{\min}+f_{\max}}{2}.
\end{equation}
Using the reference impedance \(Z_0\), the reference inductance and capacitance are defined as
\begin{equation}
L_{\mathrm{ref}}
=
\frac{Z_0}{2\pi f_{\mathrm{c}}},
\qquad
C_{\mathrm{ref}}
=
\frac{1}{2\pi f_{\mathrm{c}}Z_0}.
\end{equation}
These reference values correspond to inductive and capacitive reactance on the order of \(Z_0\) around \(f_{\mathrm c}\), and therefore provide physically reasonable component scales for initialization. Hereafter, the superscript $(0)$ denotes an initial value used to construct a starting point for the nonlinear optimization. The nonresonant elements in each Foster second-form branch are initialized as scaled reference values:
\begin{equation}
C_{\infty,ij}^{(0)}
=
\alpha_{\infty} C_{\mathrm{ref}},
\qquad
L_{0,ij}^{(0)}
=
\alpha_{0} L_{\mathrm{ref}},
\end{equation}
where \(\alpha_{\infty}\) and \(\alpha_{0}\) are prescribed scaling factors. For the \(m\)-th resonant section, the initial inductance is set to
\begin{equation}
L_{m,ij}^{(0)}
=
\alpha_{m}^{L} L_{\mathrm{ref}}, \quad m=1,\dots,K,
\end{equation}
where \(\alpha_{m}^{L}\) is selected from a prescribed set of inductance scaling factors. Given an initial resonant frequency \(f_m\), the corresponding capacitance is determined from the resonance condition as
\begin{equation}
C_{m,ij}^{(0)}
=
\frac{1}
{
\left(2\pi f_m\right)^2 L_{m,ij}^{(0)}
}, \quad m=1,\dots,K.
\end{equation}

Multiple structured initial points are generated by combining different inductance scaling factors $\alpha_{m}^{L}$ with different resonant-frequency configurations $f_{m}$. The resonant frequencies are selected at representative locations of the target frequency band, allowing starting points to explore diverse frequency-selective response patterns while maintaining physically reasonable component scales. 

The proposed multi-start gradient-based optimization is implemented using \textit{fmincon} from the MATLAB Optimization Toolbox. For each structured initial point, \textit{fmincon} is applied to minimize the negative achievable rate, \(-R_{\mathrm{MISO}}(\boldsymbol{\theta})\), subject to the box constraints
in \eqref{circuit_cs1} and \eqref{circuit_cs2}, which are imposed directly
on the logarithmic circuit variables. The solution with the highest achievable rate over all runs is retained. To reduce the computational cost and avoid finite-difference approximations, analytical gradients derived by applying the chain rule to the forward mapping established in Section~\ref{section_5_1} are supplied to the solver. The detailed gradient expressions are provided in Appendix~\ref{appendix_gradient}. 

\section{Numerical Results and Discussion}\label{section_6}
In this section, we first use a multiport conjugate-matching example to validate the proposed directional-gain interpretation, particularly the insight that better alignment with the power-normalized propagation vector can enhance the equivalent channel gain. Then, a strongly mismatched antenna array is considered to evaluate the achievable-rate upper bound under multiport Bode--Fano constraints and the proposed circuit
realization method.
\subsection{Application of Directional Gain Interpretation}\label{section_6_1}
At the center frequency, we consider a multiport conjugate-matching example to examine how the directional structure of the MMN-induced transfer matrix affects the SNR beyond the overall transmission coefficient. The simulation employs a uniform linear array of four half-wavelength dipoles. The center frequency is $f_c=5$~GHz, and the corresponding wavelength is $\lambda=\frac{c}{f_{\mathrm c}}$, where $c$ is the speed of light. The dipole length and width are set to $l=0.5\lambda$ and $w=0.05l$, respectively, and the inter-element spacing is $0.4\lambda$. The S-parameters of the mutually coupled transmit antenna array is obtained using the MATLAB Antenna Toolbox. Based on these S-parameters, a passive and lossless MMN is synthesized using the method in \cite{nie2014systematic} to achieve multiport conjugate matching at \(f_{\mathrm c}\). The transmitter--receiver distance is set to $d=800$~m. The total transmit power over the considered bandwidth $B$ is fixed at $0.01$~W, and is uniformly distributed over frequency, such that $P_{\mathrm T}(f)=\frac{0.01}{B}$. The effective noise temperature is set to $T_n=290$~K. The wireless propagation channel $\mathbf{s}_\mathrm{RT}(f)$ is generated according to the scattering model in \cite{deshpande2024generalization}. The line-of-sight path has an angle $\theta=\frac{\pi}{4}$ relative to the array broadside. In addition, eight non-line-of-sight paths are included, with $\theta_l \sim \mathcal{U}(\theta-\frac{\pi}{20},\theta+\frac{\pi}{20})$ and $\alpha_l \sim \mathcal{CN}(0,1)$.

\begin{figure}
\centering
\includegraphics[scale=0.35]{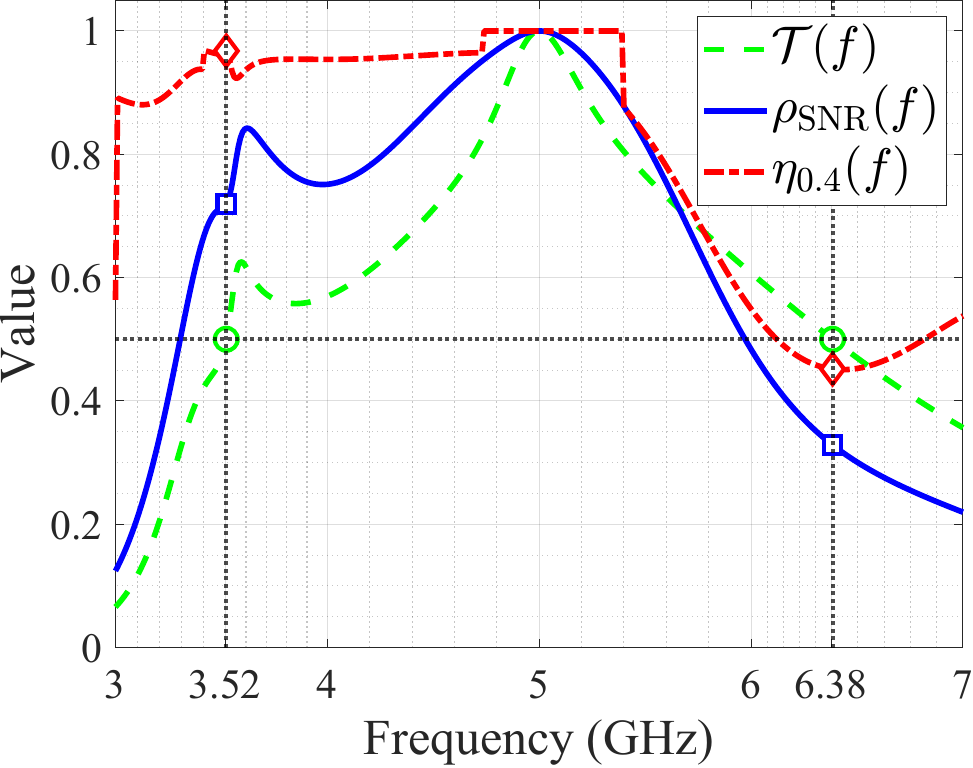}
\caption{Transmission coefficient, SNR retention factor, and alignment coefficient versus frequency for an MMN synthesized to achieve multiport conjugate matching at \(5\)~GHz.}\label{result_1_1} %图片标题
\end{figure}

Following the circuit-to-channel calculation procedure in Section~\ref{section_5_1}, the equivalent channel \(\mathbf h_{\mathrm{MISO}}(f)\) is obtained from the synthesized MMN. To quantify the resulting SNR degradation, we define the SNR retention factor as $\rho_{\mathrm{SNR}}(f)=\frac{\|\mathbf{h}_{\mathrm{MISO}}(f)\|^2}{\|\tilde{\mathbf{s}}_\mathrm{RT}(f)\|^2}$, which equals the ratio of the actual SNR to the corresponding ideal-matching SNR. Fig.~\ref{result_1_1} shows $\rho_{\mathrm{SNR}}(f)$, the transmission coefficient $\mathcal T(f)$, and the alignment coefficient $\eta_{0.4}(f)$ as functions of frequency. At the center frequency, the multiport conjugate matching network achieves perfect matching, yielding $\mathcal T(f_{\mathrm c})=1$, $\eta_{0.4}(f_{\mathrm c})=1$, and no SNR degradation occurs. Away from $f_{\mathrm c}$, $\rho_{\mathrm{SNR}}(f)$ is governed jointly by the overall transmission level and the directional alignment. Below $5$~GHz, $\eta_{0.4}(f)$ remains high despite the decrease in $\mathcal T(f)$, so the SNR degradation is relatively mild. Above $5$~GHz, both $\mathcal T(f)$ and $\eta_{0.4}(f)$ decrease, resulting in a more pronounced SNR degradation. This difference is further illustrated by comparing the results at
\(3.52\)~GHz and \(6.38\)~GHz. Although
\(\mathcal T(f)=0.5\) at both frequencies, the corresponding SNR retention factors are \(0.72\) and \(0.33\), respectively. The higher SNR retention at \(3.52\)~GHz is attributed to the stronger directional alignment, with \(\eta_{0.4}(f)=0.96\), compared with \(0.45\) at \(6.38\)~GHz.

\subsection{Bode--Fano-Constrained Rate Bound and Circuit Realization}
We next evaluate the achievable-rate upper bound and circuit realizations for the MISO system with a transmit-side MMN. The simulation employs a uniform linear array of four dipoles, with the operating bands centered at \(f_{\mathrm c}=5\)~GHz. For each bandwidth \(B\), the lower and upper band edges are \(f_{\min}=f_{\mathrm c}-\frac{B}{2}\) and \(f_{\max}=f_{\mathrm c}+\frac{B}{2}\), respectively. To induce strong mutual coupling and appreciable impedance mismatch, the antenna geometry is specified relative to the wavelength \(\lambda=\frac{c}{f_{\lambda}}\) at the design frequency \(f_{\lambda}=4.5\)~GHz. Specifically, the inter-element spacing is \(0.1\lambda\), while the dipole length and width are
\(l=0.25\lambda\) and \(w=0.008l\), respectively. The transmit power spectral density, noise temperature, and wireless propagation channel are configured in the same way as in Section~\ref{section_6_1}.

\begin{figure}
\centering
\includegraphics[scale=0.35]{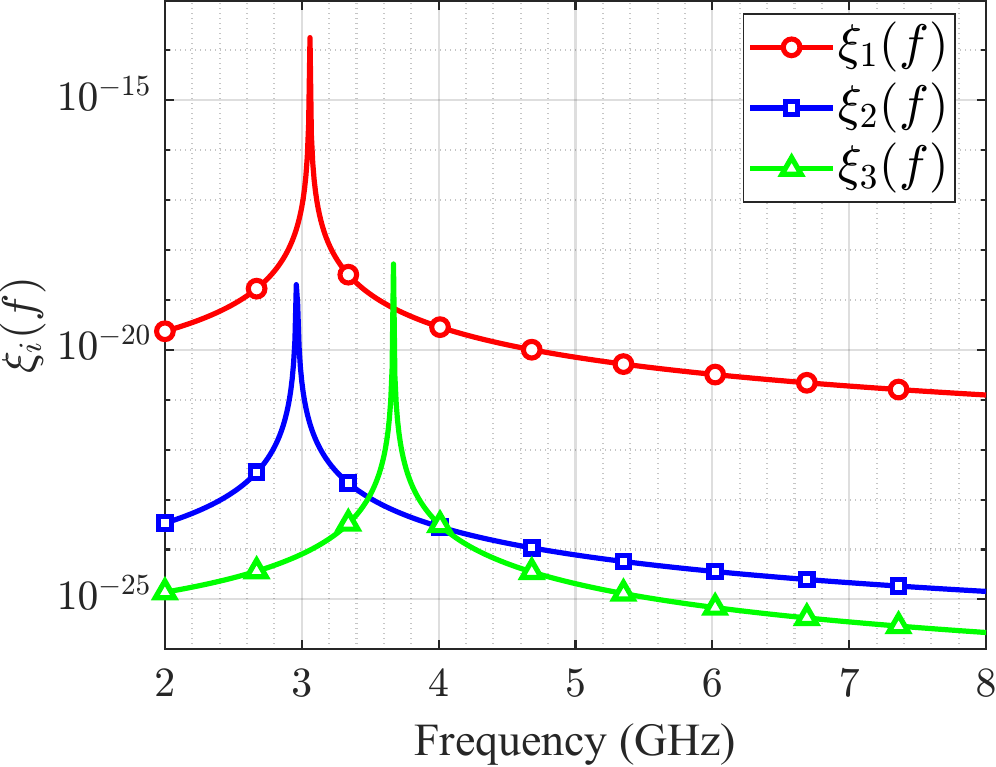}
\caption{Weighting functions \(\xi_i(f)\) of the three multiport Bode--Fano constraints for the considered four-element dipole array.}\label{BF_coefficient} %图片标题
\end{figure}

For the considered array, the multiport Bode--Fano constraints are derived from a passive rational approximation of the array scattering matrix. Specifically, the scattering matrix is first simulated over \(2\)--\(8\)~GHz at \(10\)-MHz intervals using the MATLAB Antenna Toolbox. The simulated scattering parameters are then fitted using using the open-source Matrix Fitting Toolbox~\cite{deschrijver2008macromodeling,gustavsen2009fast}, with the fitting order set to \(8\). Based on the resulting passive rational model, the weighting functions \(\xi_i(f)\) and the corresponding constraint constants \(B_{\mathrm{BF},i}\) are computed following the procedure in~\cite{nie2016bandwidth2}. The weighting functions of the three resulting Bode--Fano constraints are shown in Fig.~\ref{BF_coefficient}.

\begin{figure}
\centering
\includegraphics[scale=0.45]{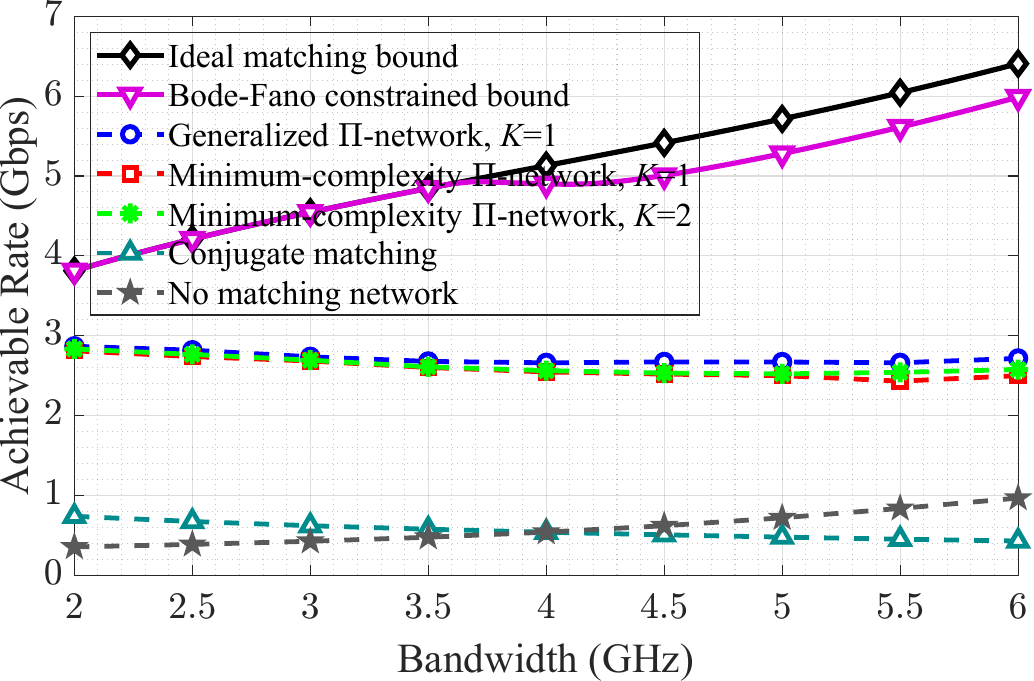}
\caption{Comparison of achievable rates among the ideal-matching bound, the Bode--Fano constrained bound, different matching-network implementations, and the unmatched case as a function of bandwidth.}\label{rate_compare} %图片标题
\end{figure}

The proposed greedy constraint-wise repair method is then used to obtain a feasible approximation to the Bode--Fano-constrained optimal transmission coefficient \(\mathcal T^\star(f)\) and the associated achievable-rate bound. For circuit implementation, three representative configurations are evaluated: minimum-complexity \(\Pi\)-network with Foster second-form orders \(K=1\) and \(K=2\), and the generalized \(\Pi\)-network with \(K=1\). For each configuration and bandwidth \(B\), the proposed multi-start gradient-based method directly optimizes the component values to maximize the achievable rate. The component bounds are set to \(C_{\min}=10^{-16}\,\mathrm{F}\), \(C_{\max}=10^{-9}\,\mathrm{F}\), \(L_{\min}=10^{-12}\,\mathrm{H}\), and \(L_{\max}=10^{-5}\,\mathrm{H}\). For all configurations, the nonresonant elements are initialized using \(\alpha_{\infty}=10^{-2}\) and \(\alpha_{0}=100\), resulting in a weakly connected initial network. Since the \(K=2\) and fully interconnected configurations involve higher-dimensional parameter spaces, a larger set
of structured initial points is used. The remaining initialization settings are specified as follows:
\begin{itemize}
    \item For the minimum-complexity \(\Pi\)-network with \(K=1\), \(\alpha_{1}^{L}\in\{1,5,10\}\) and \(f_{1}\in\{f_{\min}+0.25B,\, f_{\mathrm c},\,
    f_{\min}+0.75B\}\) are used, resulting in \(9\) initial points.

    \item For the minimum-complexity \(\Pi\) network with \(K=2\), the
    first resonant branch is initialized using
    \(\alpha_{1}^{L}\in\{1,5,10\}\) and
    \(f_{1}\in\{f_{\min}+0.2B,\, f_{\min}+0.25B,\,
    f_{\min}+0.4B,\, f_{\min}+0.6B,\,
    f_{\min}+0.75B,\, f_{\min}+0.8B\}\). The second resonant branch is initialized as a weak out-of-band branch with
    \(C_{2,ij}^{(0)}=10^{-15}\,\mathrm{F}\) and
    \(f_{2}=10f_{\max}\), while \(L_{2,ij}^{(0)}\) is determined by the resonance condition.

    \item For the generalized \(\Pi\)-network with \(K=1\), the same extended sets of \(\alpha_{1}^{L}\) and \(f_{1}\) as those
    used for the first resonant branch in the \(K=2\) configuration are adopted.
\end{itemize}

\begin{figure*}[!t]
    \centering

    % ===================== Row 1: Transmission coefficient =====================
    \subfloat[$B=2$ GHz]{
        \includegraphics[width=0.315\textwidth]{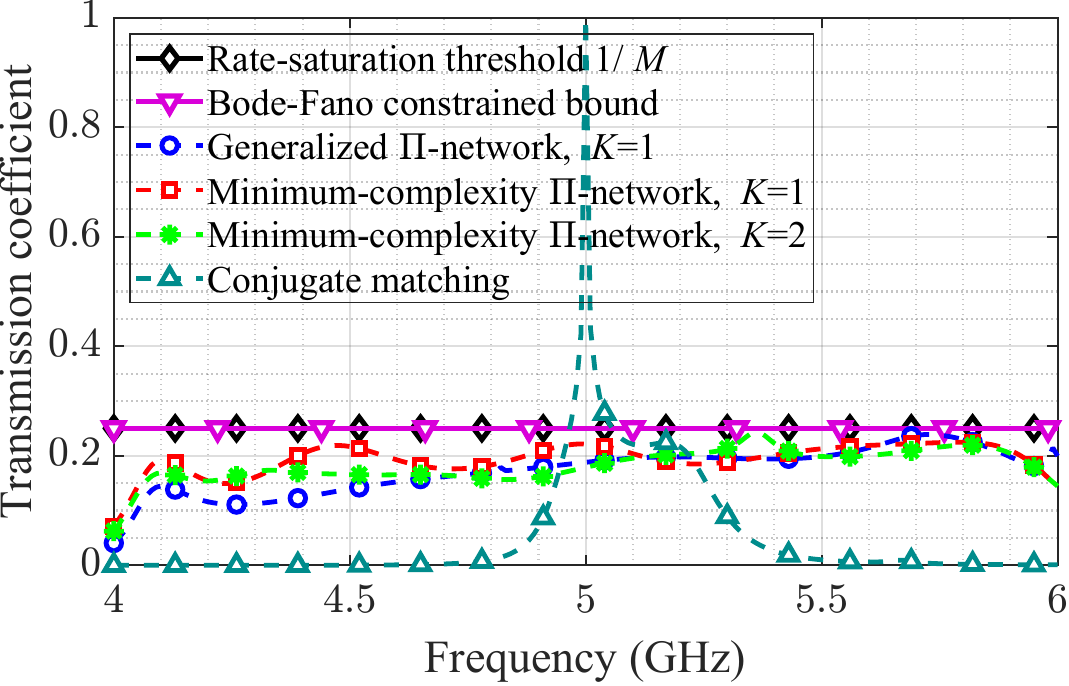}
        \label{fig:transmission_2GHz}
    }
    \hspace{0.0001\textwidth}
    \subfloat[$B=4$ GHz]{
        \includegraphics[width=0.315\textwidth]{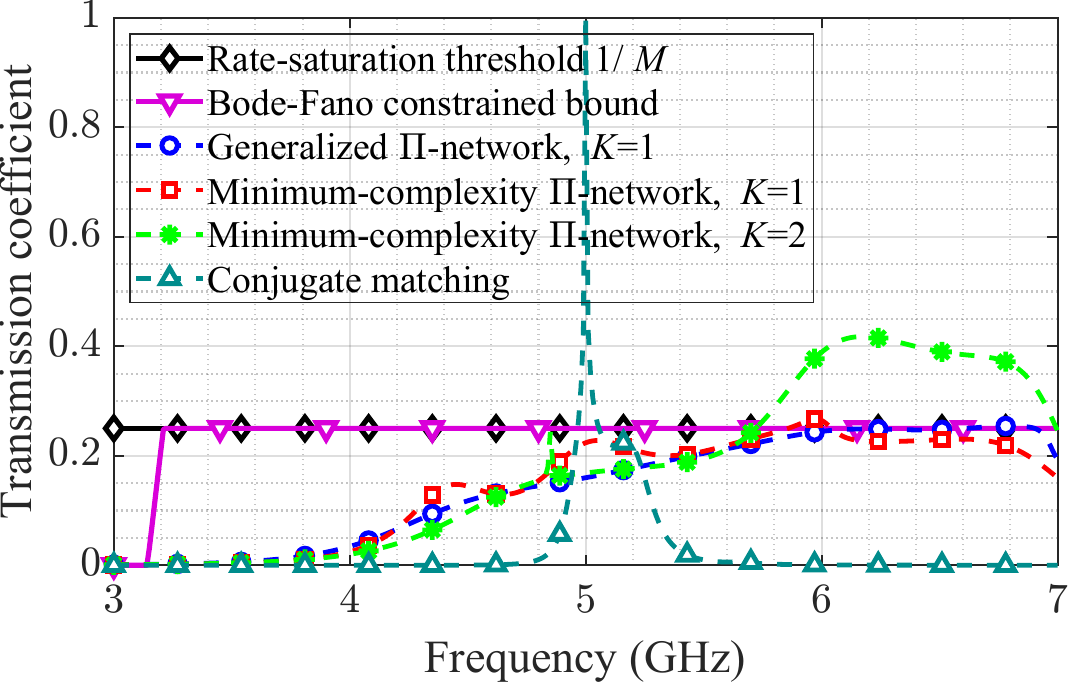}
        \label{fig:transmission_4GHz}
    }
    \hspace{0.0001\textwidth}
    \subfloat[$B=6$ GHz]{
        \includegraphics[width=0.315\textwidth]{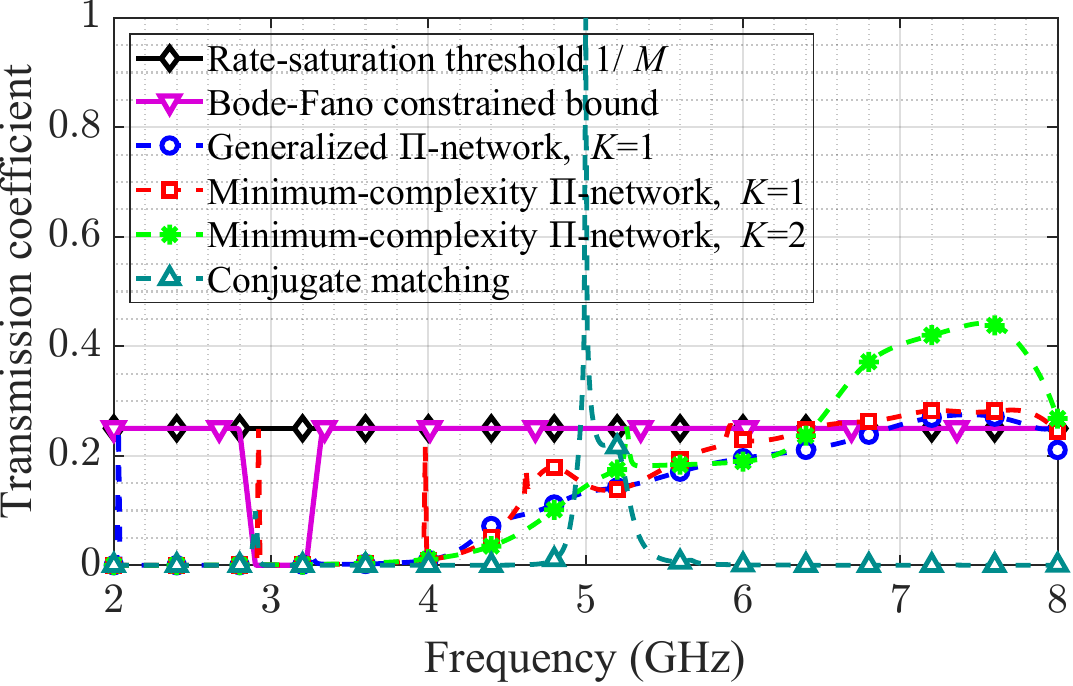}
        \label{fig:transmission_6GHz}
    }

    \vspace{0mm}

    % ===================== Row 2: Alignment coefficient =====================
    \subfloat[$B=2$ GHz]{
        \includegraphics[width=0.315\textwidth]{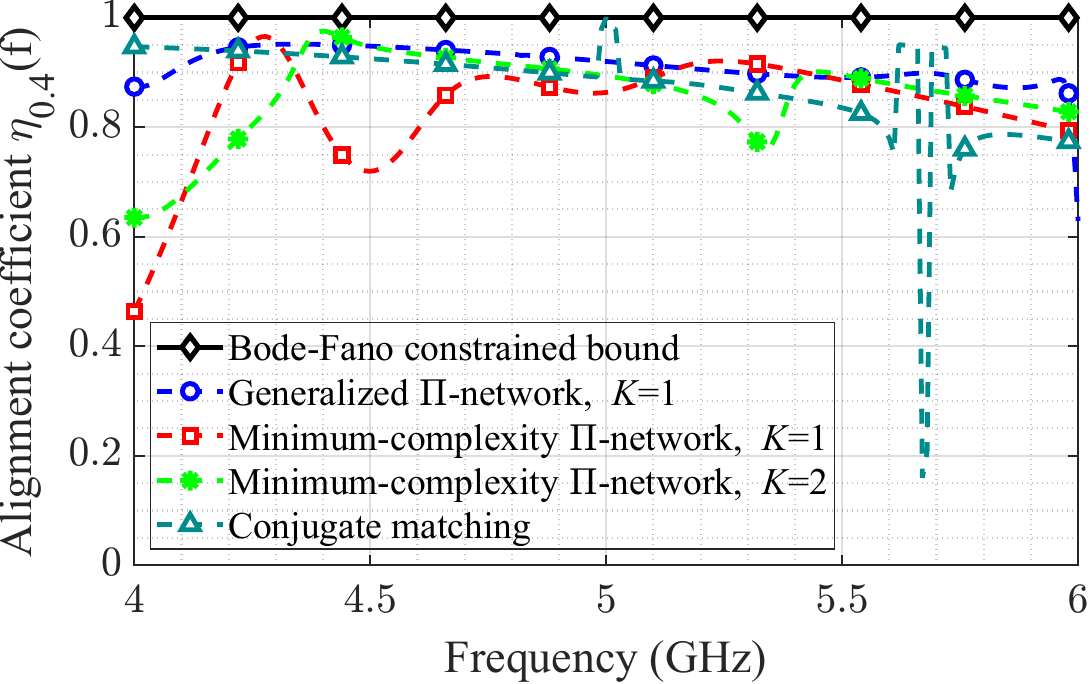}
        \label{fig:alignment_2GHz}
    }
    \hspace{0.0001\textwidth}
    \subfloat[$B=4$ GHz]{
        \includegraphics[width=0.315\textwidth]{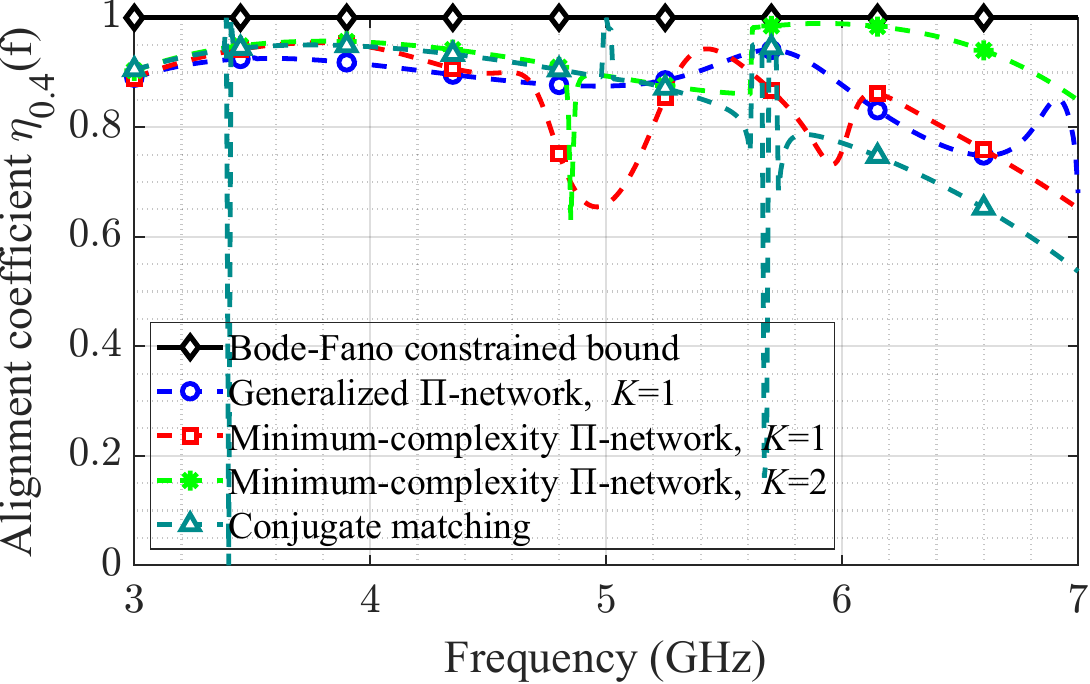}
        \label{fig:alignment_4GHz}
    }
    \hspace{0.0001\textwidth}
    \subfloat[$B=6$ GHz]{
        \includegraphics[width=0.315\textwidth]{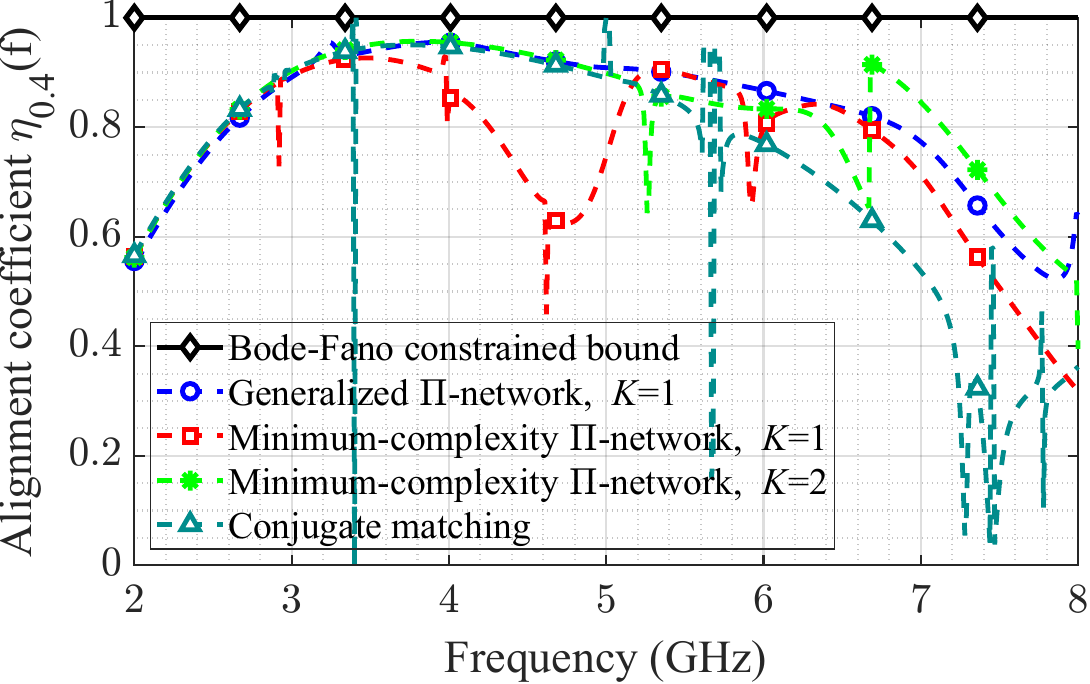}
        \label{fig:alignment_6GHz}
    }

   \caption{Frequency-dependent transmission and alignment characteristics under different bandwidths. The top row compares the Bode--Fano-constrained optimal transmission coefficient $\mathcal{T}^{\star}(f)$ with those achieved by the optimized circuit realizations, while the bottom row shows the alignment coefficient $\eta_{0.4}(f)$ of the circuit realizations. The columns correspond to $B=2$, $4$, and $6$~GHz.}
    \label{fig:transmission_alignment}
\end{figure*}

Fig.~\ref{rate_compare} compares the achievable rates of the considered MISO system over different bandwidths. The ideal-matching case assumes perfect matching throughout the operating band and therefore serves as an unattainable upper bound. When the limitations imposed by passive wideband matching are taken into account, the Bode--Fano constrained bound becomes lower, and the gap widens as the bandwidth increases. This trend indicates that the available wideband matching capability of MMN becomes increasingly insufficient to sustain the transmission coefficients required to attain the ideal equivalent channel gain across the entire band. The center-frequency conjugate matching network only improves the matching performance around \(f_{\mathrm c}\), and its broadband achievable rate therefore remains close to that of the unmatched case, for which the equivalent channel reduces to the propagation vector
\(\mathbf s_{\mathrm{RT}}(f)\). By contrast, the multiport circuits optimized directly for achievable rate significantly outperform both benchmarks, demonstrating the effectiveness of the proposed rate-oriented MMN realization method.

Among the three practical circuit implementations, the generalized \(\Pi\)-network with \(K=1\) achieves the highest rate. Its richer inter-port connectivity enables more flexible control of both the power transmission capability and the dominant transmission subspace. For the minimum-complexity \(\Pi\)-network, increasing the Foster second-form order from \(K=1\) to \(K=2\) enhances the frequency-response shaping capability of the branch admittances, resulting in a modest rate improvement. These results show that the achievable rate depends on both inter-port connectivity and Foster second-form order, which jointly govern the flexibility of the multiport transfer response.

\begin{figure*}[!t]
    \centering
    \subfloat[$B$ = 2 GHz]{
        \includegraphics[width=0.315\textwidth]{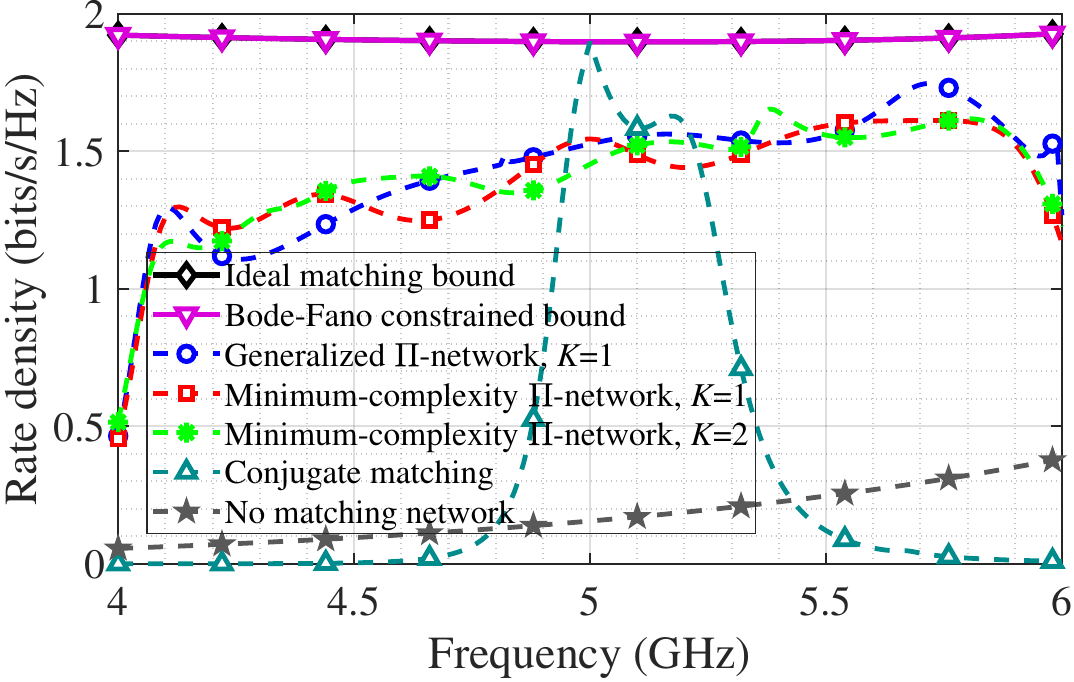}
        \label{fig:rate_density_2GHz}
    }
    \hspace{0.0001\textwidth}
    \subfloat[$B$ = 4 GHz]{
        \includegraphics[width=0.315\textwidth]{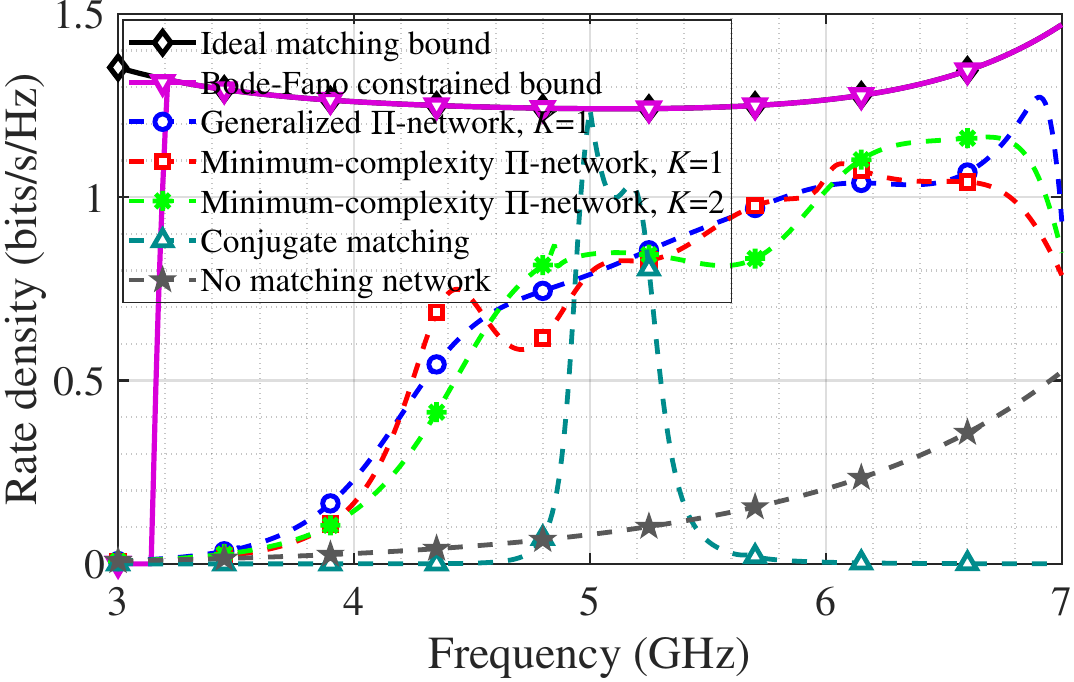}
        \label{fig:rate_density_4GHz}
    }
    \hspace{0.0001\textwidth}
    \subfloat[$B$ = 6 GHz]{
        \includegraphics[width=0.315\textwidth]{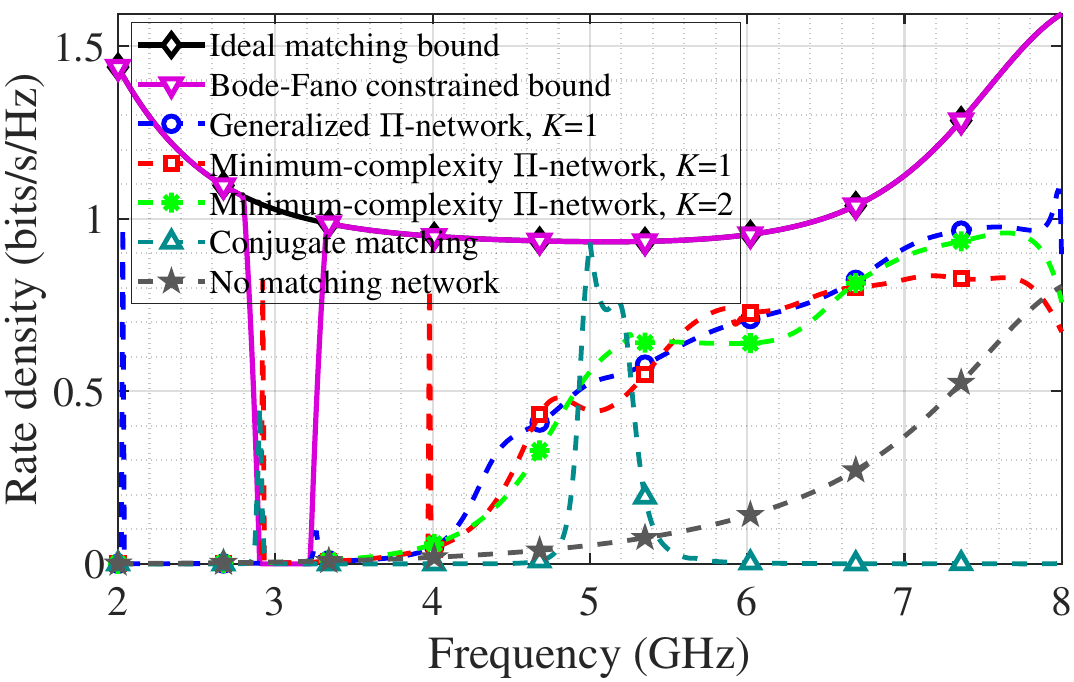}
        \label{fig:rate_density_6GHz}
    }

    \caption{Comparison of rate density among the ideal-matching bound, the Bode--Fano-constrained bound, and different matching-network implementations under different bandwidths: (a) \(B=2\) GHz, (b) \(B=4\) GHz, and (c) \(B=6\) GHz.}
    \label{fig:frequency_dependent_performance}
\end{figure*}

To further examine whether the optimized circuit responses follow the Bode--Fano constrained theoretical optimum, Fig.~\ref{fig:transmission_alignment} shows the frequency-dependent transmission coefficient \(\mathcal T(f)\) and the alignment coefficient \(\eta_{0.4}(f)\) for \(B=2\), \(4\), and \(6\)~GHz. As discussed in \textit{Remark}~\ref{remark3}, the theoretical optimum allocates the available transmission capability across frequency by balancing the achievable-rate benefit and the Bode--Fano matching cost. Moreover, as established in \textit{Remark}~\ref{remark2}, \(\mathcal T(f)=1/M\) is sufficient to attain the ideal equivalent channel gain when the power-normalized propagation vector is aligned with a unit-gain input direction of \(\mathbf F(f)\). For the considered four-element array, the corresponding saturation threshold is \(0.25\). This principle is clearly reflected in the theoretical transmission-coefficient curves. Since the weighting functions in Fig.~\ref{BF_coefficient} indicate a high matching cost around \(3\)~GHz, the theoretical optimum avoids allocating excessive transmission capability to this region when the bandwidth is increased to \(4\) or \(6\)~GHz. Consequently, \(\mathcal T^\star(f)\) exhibits a pronounced dip around these high-cost frequencies while remaining close to the saturation threshold over the lower-cost portions of the operating band.

The optimized circuit responses exhibit the same tendency, although with coarser frequency selectivity. As shown in Figs.~\ref{fig:transmission_2GHz}--\ref{fig:transmission_6GHz}, especially for larger bandwidths, the practical circuits also suppress \(\mathcal T(f)\) in the low-frequency region with a high Bode--Fano cost. Because finite-order Foster second-form circuits have limited response-shaping flexibility, they cannot exactly reproduce the sharp and localized transmission allocation of \(\mathcal T^\star(f)\). In the lower-cost frequency regions, the optimized circuits do not pursue \(\mathcal T(f)=1\). Instead, their transmission coefficients generally remain near the saturation threshold of \(0.25\). Meanwhile, Figs.~\ref{fig:alignment_2GHz}--\ref{fig:alignment_6GHz} show that the alignment coefficient remains relatively high over most of the operating band. Overall, the consistency between the theoretical optimum and the circuit-optimized responses supports the practical relevance of the Bode--Fano-constrained rate characterization. The remaining performance gap can be attributed primarily to the limited response-shaping flexibility of the finite-order Foster second-form realizations. This observation suggests that higher-order realizations or more flexible multiport topologies could enable closer approximation of the theoretical optimum,
particularly over wider bandwidths.

Finally, Fig.~\ref{fig:frequency_dependent_performance} explains why the achievable rates of the practical circuits fail to increase with bandwidth. As the bandwidth expands, the newly included low-frequency region contributes little rate density. Meanwhile, spreading the fixed total transmit power over a wider band reduces the transmit power spectral density and, consequently, the rate density in the frequency region that previously contributed most to the achievable rate. These two effects jointly prevent the achievable rate from increasing with bandwidth and even cause a slight decrease.

\section{Conclusion}\label{section_7}
In this paper, we have investigated the achievable-rate upper bound of MISO systems with transmit-side MMNs under multiport Bode--Fano constraints and developed a rate-oriented circuit realization method. Using the concept of directional gain from multivariable control theory, we have first revealed how multiport power transmission and propagation are coupled in determining the equivalent channel gain. Based on this characterization, we have then derived a tight upper bound on the equivalent channel gain in terms of the transmission coefficient, enabling the achievable-rate upper-bound problem to be reformulated as a scalar transmission-allocation problem. The resulting solution has revealed that perfect impedance matching is unnecessary for attaining the ideal-matching SNR and that limited matching resources should be allocated across frequency according to rate benefit and matching cost. Finally, we have developed a multi-start gradient-based method to optimize passive multiport circuit parameters under prescribed topologies. Numerical results have demonstrated the practical value of the theoretical analysis for communication-oriented MMN design and the effectiveness of the
proposed circuit realization method.

\appendices

\section{Proof of Proposition \ref{proposition1}}\label{appendix_A}
For notational brevity, the frequency argument $f$ is omitted. The source-side input reflection matrix is given by
\begin{equation}
\mathbf S_{\mathrm{in}}
=
\mathbf S_{\mathrm{M},11}
+
\mathbf S_{\mathrm{M},12}\mathbf S_{\mathrm T}
\left(
\mathbf I_M-\mathbf S_{\mathrm{M},22}\mathbf S_{\mathrm T}
\right)^{-1}
\mathbf S_{\mathrm{M},21},
\end{equation}
such that $\mathbf b_{\mathrm T}=\mathbf S_{\mathrm{in}}\mathbf a_{\mathrm T}$. Since the MMN is lossless, the net input power at the source side equals the net power delivered to the antenna array: $\|\mathbf a_{\mathrm T}\|^2-\|\mathbf b_{\mathrm T}\|^2
= \|\mathbf a_{\mathrm A}\|^2-\|\mathbf b_{\mathrm A}\|^2$. Using $\mathbf b_{\mathrm T}=\mathbf S_{\mathrm{in}}\mathbf a_{\mathrm T}$ and the definition of $\mathbf F$, we obtain $\mathbf a_{\mathrm T}^H \left(\mathbf I_M-\mathbf S_{\mathrm{in}}^H\mathbf S_{\mathrm{in}}\right) \mathbf a_{\mathrm T}
=\mathbf a_{\mathrm T}^H\mathbf F\mathbf F^H\mathbf a_{\mathrm T}$, for arbitrary $\mathbf{a}_{\mathrm{T}}$.
Hence,
\begin{equation}\label{eq:FF_SinSin_appendix}
\mathbf F\mathbf F^H
=
\mathbf I_{M}-\mathbf S_{\mathrm{in}}^H\mathbf S_{\mathrm{in}}.
\end{equation}
Since $\mathbf S_{\mathrm{in}}^H\mathbf S_{\mathrm{in}}\succeq \mathbf 0$, we
have $\mathbf 0\preceq \mathbf F\mathbf F^H\preceq \mathbf I_{M}$. Equivalently, $\mathbf 0\preceq\mathbf F^H\mathbf F\preceq\mathbf I_M$. Thus, all
singular values of $\mathbf F$ are no greater than one, and $\mathbf F$ is a
contraction under the Euclidean norm.

It remains to prove the equality condition. If $\|\mathbf F\mathbf x\|=\|\mathbf x\|,\, \forall\,\mathbf x$, then $\mathbf F^H\mathbf F=\mathbf I_{M}$. Since
$\mathbf F$ is square, this implies that
$\mathbf F\mathbf F^H=\mathbf I_{M}$. Comparing this result with
\eqref{eq:FF_SinSin_appendix} yields
$\mathbf S_{\mathrm{in}}^H\mathbf S_{\mathrm{in}}=\mathbf 0$, and therefore
$\mathbf S_{\mathrm{in}}=\mathbf 0$. Conversely, if $\mathbf S_{\mathrm{in}}=\mathbf 0$, then
\eqref{eq:FF_SinSin_appendix} gives $\mathbf F\mathbf F^H=\mathbf I_{M}$. Since
$\mathbf F$ is square, we also have $\mathbf F^H\mathbf F=\mathbf I_{M}$, and
therefore $\|\mathbf F\mathbf x\|=\|\mathbf x\|, \, \forall \mathbf x$. This proves
the equivalence. This completes the proof.

\section{Proof of Lemma \ref{lemma1}}\label{appendix_B}
For any nonzero \(\mathbf u\in\mathbb C^M\), according to the Rayleigh–Ritz theorem, we have
\begin{equation}
    \|\mathbf A\mathbf u\|^2
=
\mathbf u^H\mathbf A^H\mathbf A\mathbf u
\le
\lambda_{\max}(\mathbf A^H\mathbf A)\|\mathbf u\|^2 ,
\end{equation}
where $\lambda_{\max}(\cdot)$ denotes the largest eigenvalue. Since \(\mathbf A^H\mathbf A\succeq\mathbf 0\) and \(\tau=\frac{\operatorname{tr}(\mathbf{A}^{H}\mathbf{A})}{M}\), its largest eigenvalue satisfies
\begin{equation}\label{eq:AHA_bound1}
    \lambda_{\max}(\mathbf A^H\mathbf A)
\le
\operatorname{tr}(\mathbf A^H\mathbf A)
=
M \tau.
\end{equation}
Since \(\mathbf A\mathbf A^H\) and \(\mathbf A^H\mathbf A\) have the same nonzero eigenvalues, 
\(\mathbf A\mathbf A^H\preceq \mathbf I_{M}\) implies
\(\mathbf A^H \mathbf A\preceq \mathbf I_{M}\). Then we have
\begin{equation}\label{eq:AHA_bound2}
    \lambda_{\max}(\mathbf A^H \mathbf A)\le 1.
\end{equation}
According to (\ref{eq:AHA_bound1}) and (\ref{eq:AHA_bound2}), we can obtain
\begin{equation}
    \lambda_{\max}(\mathbf A^H \mathbf A)
\le
\min\{M\tau,1\},
\end{equation}
which gives
\begin{equation}
    \|\mathbf A\mathbf u\|^2
\le
\|\mathbf u\|^2
\min\{M\tau,1\}.
\end{equation}

It remains to show that the bound is tight. 
When \(0\le \tau\le \frac{1}{M}\), consider
\begin{equation}
\mathbf A
=
\sqrt{M\tau}\,
\mathbf q
\frac{\mathbf u^H}{\|\mathbf u\|},    
\end{equation}
where \(\mathbf q\in\mathbb C^M\) is an arbitrary unit-norm vector. Then,
\begin{equation}
\mathbf A\mathbf u
=
\sqrt{M\tau}\,
\mathbf q
\frac{\mathbf u^H\mathbf u}{\|\mathbf u\|}
=
\sqrt{M\tau}\,\|\mathbf u\|\mathbf q,    
\end{equation}
which yields
\begin{equation}
\|\mathbf A\mathbf u\|^2
=
M\tau\|\mathbf u\|^2
=
\|\mathbf u\|^2\min\{M\tau,1\}.    
\end{equation}
Thus, equality is attained for \(0\le \tau \le \frac{1}{M}\). Moreover,
\begin{equation}
\mathbf A\mathbf A^H
=
M\tau\,\mathbf q\mathbf q^H
\preceq \mathbf I_M,    
\end{equation}
so the constructed matrix $\mathbf A$ is feasible.

When \(\frac{1}{M}<\tau<1\), choose an integer \(p\) such that
\begin{equation}
1\le p\le \lfloor M\tau \rfloor ,    
\end{equation}
where $\lfloor x\rfloor$ denotes the greatest integer not exceeding $x$. Consider the singular value decomposition
\begin{equation}
\mathbf A
=
\mathbf U
\operatorname{diag}(\sigma_1,\ldots,\sigma_M)
\mathbf V^H,    
\end{equation}
where \(\mathbf U\) is an arbitrary unitary matrix and
\(\mathbf V=[\mathbf v_1,\ldots,\mathbf v_M]\) is a unitary matrix to be specified.
Let $\mathbf V_1=[\mathbf v_1,\ldots,\mathbf v_p]$.
Choose \(\mathbf V\) such that
\begin{equation}\label{eq:proj}
    \mathbf u=\mathbf V_1\mathbf c,
\end{equation}
where \(\mathbf c\in\mathbb C^p\) has no zero entries. Thus, \(\mathbf u\) lies in the column space of \(\mathbf V_1\), and each column of \(\mathbf V_1\) contributes to the representation of \(\mathbf u\).

The singular values are chosen to satisfy
\begin{subequations}\label{eq:singularity}
\begin{align}
    & \sigma_1=\cdots=\sigma_p=1, \\
    & 0\le\sigma_m<1, \quad m=p+1,\ldots,M, \\
    & p+\sum_{m=p+1}^{M}\sigma_m^2=M\tau .
\end{align}
\end{subequations}
Such a choice is feasible because \(p\le M\tau<M\). According to (\ref{eq:singularity}), all singular values of \(\mathbf A\) are no larger than one, and hence $\mathbf 0\preceq \mathbf A\mathbf A^H\preceq \mathbf I_{M}$. Since \(\mathbf V=[\mathbf V_1,\mathbf V_2]\) is unitary and
\(\mathbf u=\mathbf V_1\mathbf c\), we have $\mathbf V^H\mathbf u =[\mathbf c, \mathbf 0]^{T}$. Because the first \(p\) singular values are equal to one, it follows that
\begin{equation}
\operatorname{diag}(\sigma_1,\ldots,\sigma_M)\mathbf V^H\mathbf u
=
\mathbf V^H\mathbf u .    
\end{equation}
Therefore, using the unitary invariance of the Euclidean norm,
\begin{equation}
\|\mathbf A\mathbf u\|
=
\left\|
\mathbf U
\operatorname{diag}(\sigma_1,\ldots,\sigma_M)
\mathbf V^H\mathbf u
\right\|
=
\|\mathbf V^H\mathbf u\|
=
\|\mathbf u\|.    
\end{equation}
Thus,
\begin{equation}
\|\mathbf A\mathbf u\|^2
=
\|\mathbf u\|^2
=
\|\mathbf u\|^2\min\{M\tau,1\}.    
\end{equation}

For the boundary case \(\tau=1\), equality is attained by choosing $\sigma_1=\cdots=\sigma_M=1$. Thus, equality is also attained for \(\frac{1}{M}<\tau\le 1\). This completes the proof.

\section{Analytical Gradients for Circuit-Parameter Optimization}
\label{appendix_gradient}
Let \(\vartheta_{ij}\) denote an arbitrary scalar entry of 
\(\boldsymbol{\theta}_{ij}\), i.e., one optimization variable associated with branch \((i,j)\).  Unless otherwise specified, the explicit dependences on \(f\) and \(\boldsymbol{\theta}\) are omitted in the following derivative expressions for notational simplicity.

Starting from the achievable-rate expression in (\ref{eq:MISO_rate}), the derivative of \(R_{\mathrm{MISO}}(\boldsymbol{\theta})\) with respect to \(\vartheta_{ij}\) can be written as
\begin{equation}
\frac{\partial R_{\mathrm{MISO}}}
{\partial \vartheta_{ij}}
=
\int_{f_{\min}}^{f_{\max}}
\chi(f;\boldsymbol{\theta})
\operatorname{Re}
\left\{
\mathbf h_{\mathrm{MISO}}^{H}
\frac{
\partial \mathbf h_{\mathrm{MISO}}
}{
\partial \vartheta_{ij}
}
\right\}
df,
\end{equation}
where
\begin{equation}
\chi(f;\boldsymbol{\theta})
=
\frac{
2P_{\mathrm{T}}(f)
}{
\ln 2
\left[
N_{0}+P_{\mathrm{T}}(f)
\left\|
\mathbf h_{\mathrm{MISO}}(f;\boldsymbol{\theta})
\right\|^2
\right]
}.
\end{equation}
Therefore, it remains to evaluate
\(\partial \mathbf h_{\mathrm{MISO}} / \partial \vartheta_{ij}\). Differentiating the equivalent channel expression in (\ref{eq:hMISO}) gives
\begin{equation}
\begin{aligned}
\frac{\partial \mathbf h_{\mathrm{MISO}}}
{\partial \vartheta_{ij}}
=&
\left(
\frac{\partial \mathbf S_{\mathrm{M},21}}
{\partial \vartheta_{ij}}
\right)^H
\mathbf D^{-H}
\mathbf s_{\mathrm{RT}}^{*}
\\
&+
\mathbf S_{\mathrm{M},21}^{H}
\mathbf D^{-H}
\mathbf S_{\mathrm T}^{H}
\left(
\frac{\partial \mathbf S_{\mathrm{M},22}}
{\partial \vartheta_{ij}}
\right)^H
\mathbf D^{-H}
\mathbf s_{\mathrm{RT}}^{*},
\end{aligned}
\end{equation}
where
\begin{equation}
\mathbf D(f;\boldsymbol{\theta})
=
\mathbf I_{M}
-
\mathbf S_{\mathrm{M},22}(f;\boldsymbol{\theta})
\mathbf S_{\mathrm T}(f).
\end{equation}
Thus, the derivative of the equivalent channel is determined by the derivatives of the two scattering-parameter blocks
\(\mathbf S_{\mathrm{M},21}\) and \(\mathbf S_{\mathrm{M},22}\). Differentiating the \(Y\)-to-\(S\) relation in (\ref{eq:Y_to_S}) with respect to \(\vartheta_{ij}\) yields
\begin{equation}
\frac{\partial \mathbf S_{\mathrm M}}
{\partial \vartheta_{ij}}
=
-2Z_0
\left(
\mathbf I_{2M}
+
Z_0\mathbf Y
\right)^{-1}
\frac{\partial \mathbf Y}
{\partial \vartheta_{ij}}
\left(
\mathbf I_{2M}
+
Z_0\mathbf Y
\right)^{-1}.
\end{equation}
The required blocks
\(\partial \mathbf S_{\mathrm{M},21}/\partial \vartheta_{ij}\)
and
\(\partial \mathbf S_{\mathrm{M},22}/\partial \vartheta_{ij}\)
are obtained by extracting the corresponding blocks from \(\partial \mathbf S_{\mathrm M}/\partial \vartheta_{ij}\). From the admittance-matrix assembly rule in \eqref{eq:Y_stamping}, only the admittance of branch \((i,j)\) depends on the variable \(\vartheta_{ij}\). Hence,
\begin{equation}
\frac{\partial \mathbf Y}
{\partial \vartheta_{ij}}
=
\frac{\partial c_{ij}}
{\partial \vartheta_{ij}}
\mathbf E_{ij}.
\end{equation}
The remaining task is therefore reduced to the local derivative of the Foster second-form branch admittance with respect to each scalar entry of \(\boldsymbol{\theta}_{ij}\). From \eqref{eq:foster_branch_admittance}, the local derivatives are given by
\begin{equation}
\frac{\partial c_{ij}}
{\partial \theta_{ij,\infty}^{C}}
=
\mathrm{j}2\pi f C_{\infty,ij},
\end{equation}
\begin{equation}
\frac{\partial c_{ij}}
{\partial \theta_{ij,0}^{L}}
=
-\frac{1}{\mathrm{j}2\pi f L_{0,ij}},
\end{equation}
and, for \(m=1,\ldots,K\),
\begin{equation}
\frac{\partial c_{ij}}
{\partial \theta_{ij,m}^{L}}
=
-\frac{\mathrm{j}2\pi f L_{m,ij}}
{\left( \mathrm{j}2\pi f L_{m,ij}
+
\frac{1}{\mathrm{j}2\pi f C_{m,ij}}\right)^2},
\end{equation}
\begin{equation}
\frac{\partial c_{ij}}
{\partial \theta_{ij,m}^{C}}
=
\frac{1}
{\mathrm{j}2\pi f C_{m,ij}\left( \mathrm{j}2\pi f L_{m,ij}
+
\frac{1}{\mathrm{j}2\pi f C_{m,ij}}\right)^2}.
\end{equation}

These expressions complete the analytical gradient of \(R_{\mathrm{MISO}}(\boldsymbol{\theta})\) with respect to all optimization variables. Since the optimizer minimizes \(-R_{\mathrm{MISO}}(\boldsymbol{\theta})\), the supplied gradient is the negative of the rate gradient derived above.

\bibliographystyle{IEEEtran}
\bibliography{main.bib}

\begin{thebibliography}{10}
\providecommand{\url}[1]{#1}
\csname url@samestyle\endcsname
\providecommand{\newblock}{\relax}
\providecommand{\bibinfo}[2]{#2}
\providecommand{\BIBentrySTDinterwordspacing}{\spaceskip=0pt\relax}
\providecommand{\BIBentryALTinterwordstretchfactor}{4}
\providecommand{\BIBentryALTinterwordspacing}{\spaceskip=\fontdimen2\font plus
\BIBentryALTinterwordstretchfactor\fontdimen3\font minus \fontdimen4\font\relax}
\providecommand{\BIBforeignlanguage}[2]{{%
\expandafter\ifx\csname l@#1\endcsname\relax
\typeout{** WARNING: IEEEtran.bst: No hyphenation pattern has been}%
\typeout{** loaded for the language `#1'. Using the pattern for}%
\typeout{** the default language instead.}%
\else
\language=\csname l@#1\endcsname
\fi
#2}}
\providecommand{\BIBdecl}{\relax}
\BIBdecl

\bibitem{cui2023absorptive}
J.~Cui, L.~Zhu, and H.~Deng, ``Absorptive filtering antennas with predictable flat gain and wideband impedance matching,'' \emph{IEEE Trans. Antennas Propag.}, vol.~72, no.~3, pp. 2382--2390, Mar. 2024.

\bibitem{cheng2024multimode}
S.~M. Cheng and D.~Psychogiou, ``Multimode low noise amplifier filter with tunable transfer function characteristics,'' \emph{IEEE Trans. Microw. Theory Techn.}, vol.~73, no.~3, pp. 1487--1501, 2025.

\bibitem{soltani2024wide}
M.~Soltani and G.~V. Eleftheriades, ``Wide-angle impedance matching of a patch antenna phased array using artificial dielectric sheets,'' \emph{IEEE Trans. Antennas Propag.}, vol.~72, no.~5, pp. 4258--4270, May. 2024.

\bibitem{lau2006impact}
B.~K. Lau, J.~B. Andersen, G.~Kristensson, and A.~F. Molisch, ``Impact of matching network on bandwidth of compact antenna arrays,'' \emph{IEEE Trans. Antennas Propag.}, vol.~54, no.~11, pp. 3225--3238, Nov. 2006.

\bibitem{abdallah2016maximum}
M.~N. Abdallah, T.~K. Sarkar, and M.~Salazar-Palma, ``Maximum power transfer versus efficiency,'' in \emph{Proc. IEEE Int. Symp. Antennas Propag. (APSURSI)}, Jun. 2016, pp. 183--184.

\bibitem{putra2015impedance}
A.~P. Putra and A.~Munir, ``{Impedance matching circuit for 30--88MHz monopole antenna with automatic matching},'' in \emph{Proc. 1st Int. Conf. Wireless Telematics (ICWT)}, Nov. 2015, pp. 1--4.

\bibitem{choi2020adaptive}
J.~Choi, D.~Choi, J.~Lee, W.~Hwang, and W.~Hong, ``{Adaptive 5G architecture for an mmWave antenna front-end package consisting of tunable matching network and surface-mount technology},'' \emph{IEEE Trans. Compon., Packag., Manuf. Technol.}, vol.~10, no.~12, pp. 2037--2046, Dec. 2020.

\bibitem{rahola2011optimization}
J.~Rahola, ``Optimization of matching circuits for antennas,'' in \emph{Proc. 5th Eur. Conf. Antennas Propag. (EUCAP)}, Apr. 2011, pp. 776--778.

\bibitem{desoer1973maximum}
C.~Desoer, ``The maximum power transfer theorem for n-ports,'' \emph{IEEE Trans. Circuit Theory}, vol.~20, no.~3, pp. 328--330, May. 1973.

\bibitem{li2021decoupling}
M.~Li \emph{et~al.}, ``{Decoupling and matching network for dual-band MIMO antennas},'' \emph{IEEE Trans. Antennas Propag.}, vol.~70, no.~3, pp. 1764--1775, Mar. 2022.

\bibitem{broyde2014some}
F.~Broyd{\'e} and E.~Clavelier, ``Some properties of multiple-antenna-port and multiple-user-port antenna tuners,'' \emph{IEEE Trans. Circuits Syst. I, Reg. Papers}, vol.~62, no.~2, pp. 423--432, Feb. 2015.

\bibitem{nie2014systematic}
D.~Nie, B.~M. Hochwald, and E.~Stauffer, ``Systematic design of large-scale multiport decoupling networks,'' \emph{IEEE Trans. Circuits Syst. I, Reg. Papers}, vol.~61, no.~7, pp. 2172--2181, Mar. 2014.

\bibitem{shen2015impedance}
S.~Shen and R.~D. Murch, ``{Impedance matching for compact multiple antenna systems in random RF fields},'' \emph{IEEE Trans. Antennas Propag.}, vol.~64, no.~2, pp. 820--825, Feb. 2016.

\bibitem{sarabandi2022bandwidth}
K.~Sarabandi and M.~Rao, ``{Bandwidth and SNR of small receiving antennas: To match or not to match},'' \emph{IEEE Trans. Antennas Propag.}, vol.~71, no.~1, pp. 99--104, Jun. 2023.

\bibitem{li2022study}
Z.~Li \emph{et~al.}, ``{A study on electromagnetic scattering characteristics of 4-D antenna arrays},'' \emph{IEEE Trans. Antennas Propag.}, vol.~71, no.~1, pp. 275--287, Jan. 2023.

\bibitem{ivrlavc2010toward}
M.~T. Ivrla{\v{c}} and J.~A. Nossek, ``Toward a circuit theory of communication,'' \emph{IEEE Trans. Circuits Syst. I, Reg. Papers}, vol.~57, no.~7, pp. 1663--1683, Jun. 2010.

\bibitem{wallace2004mutual}
J.~W. Wallace and M.~A. Jensen, ``{Mutual coupling in MIMO wireless systems: A rigorous network theory analysis},'' \emph{IEEE Trans. Wireless Commun.}, vol.~3, no.~4, pp. 1317--1325, Jun. 2004.

\bibitem{fei2008optimal}
Y.~Fei, Y.~Fan, B.~K. Lau, and J.~S. Thompson, ``{Optimal single-port matching impedance for capacity maximization in compact MIMO arrays},'' \emph{IEEE Trans. Antennas Propag.}, vol.~56, no.~11, pp. 3566--3575, Nov. 2008.

\bibitem{saab2021optimizing}
S.~Saab, A.~Mezghani, and R.~W. Heath, ``Optimizing the mutual information of frequency-selective multi-port antenna arrays in the presence of mutual coupling,'' \emph{IEEE Trans. Commun.}, vol.~70, no.~3, pp. 2072--2084, Mar. 2022.

\bibitem{kundu2017enhancing}
L.~Kundu and B.~L. Hughes, ``{Enhancing capacity in compact MIMO-OFDM systems with frequency-selective matching},'' \emph{IEEE Trans. Commun.}, vol.~65, no.~11, pp. 4694--4703, Nov. 2017.

\bibitem{fano1950theoretical}
R.~M. Fano, ``Theoretical limitations on the broadband matching of arbitrary impedances,'' \emph{J. Franklin Inst.}, vol. 249, no.~1, pp. 57--83, Jan. 1950.

\bibitem{deshpande2023achievable}
N.~V. Deshpande \emph{et~al.}, ``{Achievable rate of a SISO system under wideband matching network constraints},'' in \emph{Proc. IEEE Global Commun. Conf.}, Dec. 2023, pp. 7514--7519.

\bibitem{deshpande2024generalization}
N.~V. Deshpande \emph{et~al.}, ``{A generalization of the achievable rate of a MISO system using Bode-Fano wideband matching theory},'' \emph{IEEE Trans. Wireless. Commun.}, vol.~23, no.~10, pp. 13\,313--13\,329, Oct. 2024.

\bibitem{barneto2022beamformer}
C.~B. Barneto \emph{et~al.}, ``{Beamformer design and optimization for joint communication and full-duplex sensing at mm-Waves},'' \emph{IEEE Trans. Commun.}, vol.~70, no.~12, pp. 8298--8312, Dec. 2022.

\bibitem{nie2016bandwidth}
D.~Nie and B.~M. Hochwald, ``{Bandwidth analysis of multiport radio frequency systems—Part I},'' \emph{IEEE Trans. Antennas Propag.}, vol.~65, no.~3, pp. 1081--1092, Mar. 2017.

\bibitem{skogestad2005multivariable}
S.~Skogestad and I.~Postlethwaite, \emph{Multivariable Feedback Control: Analysis and Design}.\hskip 1em plus 0.5em minus 0.4em\relax Hoboken, NJ, USA: Wiley, 2007.

\bibitem{nie2016bandwidth2}
D.~Nie and B.~M. Hochwald, ``{Bandwidth analysis of multiport radio-frequency systems—Part II},'' \emph{IEEE Trans. Antennas Propag.}, vol.~65, no.~3, pp. 1093--1107, Mar. 2017.

\bibitem{rodriguez2021systematic}
R.~Rodriguez-Berral, F.~Mesa, and F.~Medina, ``Systematic obtaining of foster’s equivalent circuits for symmetric frequency selective surfaces,'' \emph{IEEE Trans. Antennas Propag.}, vol.~70, no.~2, pp. 1166--1177, Feb. 2022.

\bibitem{hao2022realizable}
L.~Hao and G.~Shi, ``{Realizable reduction of Multi-Port RCL networks by block elimination},'' \emph{IEEE Trans. Circuits Syst. I, Reg. Papers}, vol.~70, no.~1, pp. 399--412, Jan. 2023.

\bibitem{deschrijver2008macromodeling}
D.~Deschrijver, M.~Mrozowski, T.~Dhaene, and D.~De~Zutter, ``Macromodeling of multiport systems using a fast implementation of the vector fitting method,'' \emph{IEEE Microw. Wireless Compon. Lett.}, vol.~18, no.~6, pp. 383--385, Jun. 2008.

\bibitem{gustavsen2009fast}
B.~Gustavsen, ``{Fast passivity enforcement for S-parameter models by perturbation of residue matrix eigenvalues},'' \emph{IEEE Trans. Adv. Packag.}, vol.~33, no.~1, pp. 257--265, Feb. 2010.

\end{thebibliography}

\end{document}